\documentclass[12pt, a4paper]{article}
\usepackage[left=1.25in, top=1in, bottom=1.5in, right=1.25in]{geometry}
\usepackage[utf8]{inputenc} 
\usepackage[T1]{fontenc}    
\usepackage[colorlinks=true,citecolor=blue]{hyperref}       
\usepackage{url}            
\usepackage{booktabs}       
\usepackage{amsfonts}       
\usepackage{nicefrac}       
\usepackage{microtype}      
\usepackage{bbm}
\usepackage{natbib}
\usepackage{float}
\usepackage[font=footnotesize,format=hang]{caption}
\usepackage{graphicx}
\usepackage{xcolor}
\usepackage[export]{adjustbox}
\usepackage{enumerate}
\usepackage{amsmath,amssymb,amsthm}
\usepackage[mathscr]{eucal}
\usepackage{mdframed}
\usepackage{setspace}
\usepackage{booktabs,colortbl,multirow}
\usepackage{dsfont}
\usepackage{subcaption}

\graphicspath{{figs//}}
\allowdisplaybreaks

\usepackage{enumitem}
\setlist{nolistsep}
\usepackage{mwe}

\newcommand{\dcb}{\begin{array}{lll}}
\newcommand{\dce}{\end{array}}
\newcommand{\ebe}{\begin{enumerate}\setlength{\baselineskip}{13pt}\setlength{\parskip}{7pt}}
\newcommand{\dbe}{\end{enumerate} \vspace{3pt}}

\newcommand{\ibegin}{\begin{itemize}\setlength{\baselineskip}{19pt}\setlength{\parskip}{7pt}}
\newcommand{\iend}{\end{itemize}}
\newcommand{\ok}{\rule{5pt}{7pt}}

\newtheorem{Theorem}{Theorem}[section]
\newtheorem {Cor}[Theorem]{Corollary}
\newtheorem {definition}[Theorem]{Definition}
\newtheorem {Lemma}[Theorem]{Lemma}
\newtheorem {rem}[Theorem]{Remark}
\newtheorem {assumption}[Theorem]{Assumption}

\newcommand {\bd}{\begin{definition}}
\newcommand {\ed}{\end{definition}}
\newcommand {\bl}{\begin{Lemma}}
\newcommand {\el}{\end{Lemma}}
\newcommand {\bcor}{\begin{Cor}}
\newcommand {\ecor}{\end{Cor}}
\newcommand {\brem }{\begin{rem} \rm }
\newcommand {\erem }{\end{rem}}
\newcommand{\bethe}{\begin{Theorem}}
\newcommand{\ethe}{\end{Theorem}}
\newcommand {\bassumption}{\begin{assumption}}
\newcommand {\eassumption}{\end{assumption}}

\newtheorem{corollary}{Corollary}
\newtheorem{lemma}{Lemma}

\newcommand{\p}{{p}}

\newcommand{\ymis}{Y^{\textrm{mis}}_i}
\newcommand{\ymishat}{\hat{Y}^{\textrm{mis}}_i}
\newcommand{\yobs}{Y^{\textrm{obs}}_i}
\newcommand{\Ymis}{\textbf{Y}^{\textrm{mis}}}
\newcommand{\Yobs}{\textbf{Y}^{\textrm{obs}}}
\newcommand{\W}{\textbf{W}}
\newcommand{\X}{\textbf{X}}

\newcommand{\yzero}{Y_i(0)}
\newcommand{\yone}{Y_i(1)}
\newcommand{\varbeta}{\sigma^2_{\beta}}
\newcommand{\betac}{\beta^{{[c]}}}
\newcommand{\betat}{\beta^{{[t]}}}
\newcommand{\x}{\textbf{x}_i}
\newcommand{\xic}{\x^\top\betac}
\newcommand{\xit}{\x^\top\betat}
\newcommand{\betaall}{\boldsymbol{\beta}}
\newcommand{\epsall}{\boldsymbol{\epsilon}}

\newcommand{\thetaall}{\boldsymbol\vartheta}
\newcommand{\et}{{\epsilon_i^{[t]}}}
\newcommand{\ec}{{\epsilon_i^{[c]}}}
\newcommand{\epstilde}{\tilde{\epsilon}}
\newcommand{\mumis}{\mu_i^{\textrm{mis}}}
\newcommand{\muobs}{\mu_i^{\textrm{obs}}}

\newcommand{\xiobs}{\xi_i^{\textrm{obs}}}
\newcommand{\ximis}{\xi_i^{\textrm{mis}}}
\newcommand{\xitilde}{\tilde{\textbf{x}}^{\textrm{obs}}_i}
\newcommand{\xitildemis}{\tilde{\textbf{x}}^{\textrm{mis}}_i}
\newcommand{\xtilde}{\tilde{\textbf{x}}}
\newcommand{\Xtilde}{\tilde{\textbf{X}}}
\newcommand{\Rtilde}{\tilde{\textbf{R}}}

\title{Bayesian causal inference for count potential outcomes}

\author{Young Lee\\ \small Harvard University 
        \and 
        Wicher P. Bergsma\\ \small London School of Economics 
        \and 
        Marie-Ab{\`e}le Bind\\ \small Harvard University}

\date{\today}

\usepackage{soul}

\begin{document}

\maketitle

\begin{abstract}
The literature for count modeling provides useful tools to conduct causal inference when outcomes take non-negative integer values. Applied to the potential outcomes framework, we link the Bayesian causal inference literature to statistical models for count data. We discuss the general architectural considerations for constructing the predictive posterior of the missing potential outcomes. Special considerations for estimating average treatment effects are discussed, some generalizing certain relationships and some not yet encountered in the causal inference literature.
\end{abstract}

\section{Introduction}
Statistical analyses  with non-negative integer outcomes are encountered in many fields. For example, studies examining risk factors for seizure counts \citep{Burneo:2008}; number of deaths \citep{Baccini:2017}; counts of COVID-19 deaths \citep{Dominici:2020}, number of drinks over a period of time \citep{Horton:2007}. These outcomes have a particular domain, the set of non-negative integer values, $\mathbb{Z}_+$.

There is an extensive literature on statistical models for count data. Among some of the contributions to count regression model, the authors \cite{el1973bayesian,Lawless:1987,Diggle:1998,Winkelmann:2008,Chan2012CountingPW,KimSunduk:2013} study the estimation procedures of Poisson and negative binomial regression paradigms. \cite{Breslow:1984,Agresti:2007} introduced the lognormal-Poisson regression and their variants. These models have mostly been used to directly regress the observed outcomes on the observed treatment and background covariates, particularly in environmental epidemiology \citep{Gasparrini:2009,Zigler:2014,Schwartz:2015}, as opposed to causal modeling. Leveraging on the existing statistical literature, we link the Bayesian causal inference literature with models for count data. 
The fundamental problem of causal inference is one of missing data, and specifically of missing potential outcomes. One approach to handle the missing data problem of causal inference is multiple imputation, i.e., `fill in' missing data with credible values \citep{Rubin:1978,rubin:1987-buku} by repeatedly sampling from the predictive posterior of the missing potential outcomes. 

The idea of using a count distribution to model the potential outcomes is not new; \cite{Gutman:2017} proposed a procedure for estimating the causal effects of nursing home bed-hold policy and assumed a Poisson distribution for the potential outcomes. \cite{Sommer:2018} examined the causal effects of heat and rain on the number of crimes in Boston assuming directly a negative binomial distribution for the missing potential outcomes. 
Furthermore, these examples draw Bayesian inferences, which can be computationally intensive for large data sets. If the count model yields a likelihood function that is computationally expensive, the Bayesian perspective can quickly become prohibitive. One Markov chain Monte carlo (MCMC) method that is suited to imputation problems is the two-step process of data augmentation algorithm by  \cite{Tanner:1987}. However, for more realistic count models, the second step is typically intractable, and at least a two step MCMC procedure is still required to draw missing values, confer e.g., Chapter 8.3 of \cite{GelmanHill:2007}. Several threads of research have been devoted to examine the strategies for approximating draws in an imputation model, which include the large sample approximations and importance sampling \citep{rubin:1987-buku}. However, almost all these studies inherently assume continuity of the potential outcome distributions \citep{Schafer:1997,Shafer:2002}. To the best of our knowledge, little or no work has appeared on approximation strategies for an imputation model with count potential outcomes. 

In this paper, we initiate a framework for estimating causal inference within a Bayesian setting when the potential outcomes are counts. Although it may be possible to draw causal inferences using continuous or binary potential outcomes \citep{Rubin2006,hill2011bayesian,GutmanRubin:2012,Gutman-Rubin:2015,ImbensGuidoRubin:2015}, the paradigms of the aforementioned models are not generally suitable for count potential outcomes.
Here, we posit that the potential outcomes inherit statistical properties of counts (i.e., discrete non-negative integers) and can be characterized with distributions that account for overdispersion (e.g., the variance can be greater than the mean), which can occur in settings with
heterogeneous units or dependence between events. We remark that our proposed framework is flexible and can easily be parameterized to exploit a wide range of existing properties of count models, thereby could be adapted to each causal inference application. In addition to proposing a class of count potential outcomes, we add some new aspects to the existing theory by introducing some approximation strategies for drawing Bayesian causal inferences. Our approximation operates orders of magnitude faster than exact Hamiltonian MCMC (HMC). We derive an asymptotic expression for the accuracy of the approximation in terms of the total variation distance to the true posterior. Simulations show good finite sample performance. 

Determining the causal effects in this setting involves the following key steps: $(i)$ posit a suitable family of distributions for the potential outcomes and derive the conditional distribution of the missing potential outcomes given the observed data; $(ii)$ compute the posterior distribution of the parameters of the potential outcomes. The posterior distribution of Poisson parameters with Gaussian priors is intractable. However, we show a normal approximation to the posterior can be given and we derive an asymptotic rate of convergence of the approximation to the posterior in terms of the total variation distance; $(iii)$ we evaluate the conditional distribution of the missing data given the observed data, and finally $(iv)$ the estimand of causal interest is immediate. 

The paper is structured in the following manner. In Section \ref{sec:count-potential-outcomes} we introduce our non-negative potential outcomes framework. In Section \ref{sec:the-four-steps}, we develop a Bayesian imputation model for the missing potential outcomes of count type. Of particular interest is the characterization of approximation that is used to compute the posterior distribution towards the evaluation of causal estimands. Section \ref{sec:experiments} presents some numerical illustrations and Section \ref{sec:postlude} concludes.





\section{A Bayes model for count potential outcomes}
\label{sec:count-potential-outcomes}

The concept of potential outcomes was launched in \cite{splawa1990application} for Neymanian inference in randomized experiments and later used by other researchers including \cite{Kempthorne:1955} and \cite{Wilk:1955} for causal inference from randomized experiments. The concept was extended by  \cite{Rubin:1974,Rubin:1975,rubin:1976a,Rubin:1977,Rubin:1978} to other forms of causal inference from randomized experiments to observational studies \citep{Rubin:1974}. 

Following \cite{Rubin:1978}, $W_i$ denotes the treatment assignment for the $i^{th}$ unit, with $i=1,\hdots,N$, where $W_i=1$ indicates the active treatment and $W_i=0$ the control treatment. We also denote the potential outcomes had unit $i$ been assigned to active and control treatment by $\yone$ and $\yzero$, respectively. We define the unit-level causal effect of the binary treatment $W_i$ using the difference between $\yone$ and $\yzero$. We can never observe both $\yone$ and $\yzero$ for any unit $i$, because it is not possible to go back in time and expose the $i^\mathrm{th}$ unit to the other treatment. This is called the `fundamental problem of causal inference' \citep{Rubin:1975,Holland:1986}. Put differently, we are implicitly trying to figure out what would have happened to an individual had he/she taken the other treatment condition. Therefore, the problem of inferring unit-level causal effects is a missing data problem. The Bayesian framework permits the estimation of unit-level and average treatment effects \citep{Rubin:1978}. In this article, we focus on the finite population average treatment effect, defined as 
\begin{align}
\label{eq:taufs}
\textrm{ATE}:=\frac{1}{N}\sum_{i=1}^N(\yone-\yzero)
\end{align}
and we use the fact that
\begin{align}
\label{eq:identity}
\ymis = (1-W_i)\cdot\yone+W_i\cdot \yzero.
\end{align}

\paragraph{Non-negative integer valued potential outcomes.} We assume full compliance and the stable unit treatment value assumption (SUTVA) \citep{rubin1980randomization}, which means that the potential outcome of a particular unit depends only on the treatment combination it is assigned, rather than on the assignments of the remaining units, as well as that there are no hidden versions of the treatments not represented by the values of $W$. Under this assumption, the potential outcome $Y_i$ of each unit $i$ in an experiment only depends on whether it receives the treatment $(W=1)$ or not $(W=0)$. The potential outcomes of all $N$ units in an experiment, can be partitioned into two vectors of $N$ components: $\textbf{Y}(0)$ for all outcomes under control and $\textbf{Y}(1)$ for all outcomes under treatment. The unit-level potential outcomes can also be expressed as the observed and missing outcomes. Therefore, half of all the potential outcomes are observed, denoted by $\Yobs$; the other half are unobserved, denoted by $\Ymis$.

We propose a count model that allows for overdispersion:
	\begin{align}
	\label{eq:model}
	&Y_i(0)\,|\,\betac,\epsilon^{[c]}_i \sim \mathsf{Pois}(\mu_i^{[c]}),\quad Y_i(1)\,|\,\betat,\epsilon^{[t]}_i \sim \mathsf{Pois}(\mu_i^{[t]})
	\end{align}
	where $\mathsf{Pois}$ denotes the Poisson distribution, and
	\begin{align}\label{eq:model2}
	    \mu_i^{[c]}:=\exp ({\x^\top\betac})\epsilon^{[c]}_i, \quad \mu_i^{[t]}:=\exp ({\x^\top\betat})\epsilon^{[t]}_i.
	\end{align}
	The quantity $\epsilon^{[\cdot]}_i$ is taken to be a non-negative multiplicative random-effect term to model individual heterogeneity \citep{Long:1997,GelmanHill:2007,Winkelmann:2008,CameronTrivedi:2013,gelman2013bayesian} and their hyperparameters are denoted by $\vartheta^{[\cdot]}$. Let $\epsilon_i:=(\ec , \et)^{\top}$, $\epsall:=(\epsilon_1,\epsilon_2,\hdots,\epsilon_N)^{\top}$,
	and $\thetaall:=(\vartheta^{[c]}, \vartheta^{[t]})^{\top}$. The covariates $\x=(1,x_{i1},\hdots,x_{ik})$ are the $(k+1)$-dimensional features for every $i$ and let $\X=(\textbf{x}^{\top}_1,\hdots,\textbf{x}^{\top}_N)^{\top}$. We also have $\betac$ and $\betat\in\mathbb{R}^{k+1}$. We further assume a prior distribution $\pi(\betaall)$ for $\betaall$:
	\begin{align}
	\label{eq:prior for beta}
	&\betaall:=\begin{pmatrix}
	\betac\\
	\betat
	\end{pmatrix} \sim \mathsf{N}\left(\textbf{0}_{2(k+1)},\varbeta\cdot\mathbf{I}_{2(k+1)}\right)
	\end{align}
with $\sigma^2_{\beta}$ being a fixed positive number, $\mathsf{N}$ denoting the Gaussian distribution and $\betaall\in\mathbb{R}^{2(k+1)}$. In this paper, we examine two types of potential outcomes: (1) when $\epsall=1$, we have the Poisson potential outcomes that do not account for overdispersion in the count data; (2) we also work with a model that incorporate overdispersion, i.e., when $\epsilon^{[c]}$ and $\epsilon^{[t]}$ take lognormal priors. We call this the lognormal-Poisson potential outcomes, in concert with the terminology introduced in the context of the regression model \citep{Breslow:1984,Agresti:2007}. Specifically, we let $\epsilon^{[c]}\sim\log N(0,(\sigma^{[c]})^2)$ and 
$\epsilon^{[t]}\sim\log N(0,(\sigma^{[t]})^2)$. The hyperparameters for this model are $(\sigma^{[c]},\sigma^{[t]})$ and we place the priors
\begin{align}
\label{eq:priors_overdispersion}
    \sigma^{[c]}\sim\mathsf{IG}(\alpha^{[c]},\nu^{[c]})\,\,\text{and}\,\,\sigma^{[t]}\sim\mathsf{IG}(\alpha^{[t]},\nu^{[t]})
\end{align}
where the shorthand $\mathsf{IG}$ denotes the inverse gamma distribution that has density $\nu^{\alpha}/\Gamma(\alpha)(1/x)^{\alpha+1}\exp(-\nu/x)$
with parameters $\alpha$ and $\nu$ over the support $x>0$.
	
We remark that other forms of potential outcomes can be constructed in a similar manner, i.e., for example the negative binomial potential outcomes can be obtained by placing Gamma priors on $\epsilon^{[\cdot]}_i$ \citep{Lawless:1987,Hilbe:2007}. Another possible avenue is to attribute an inverse-Gaussian prior on $\epsilon^{[\cdot]}_i$ that would result in a heavy-tailed count behavior \citep{Dean:1989}. 

\paragraph{The assignment mechanism.} Drawing inferences for causal effects requires the specification of an assignment mechanism, that is, a probabilistic model for how experimental units are allocated to the treatment combination given the potential outcomes and covariates. Let $\W$ denote the $N$-vector of treatment assignments, with $i^{\mathrm{th}}$ element $W_i$. We define $N_c:=\sum_{i=1}^N(1-W_i)$ and $N_t=\sum_{i=1}^N W_i$ as the number of units assigned to the control and active treatment respectively, with $N_c+N_t=N$. We assume a completely randomized assignment mechanism, so by definition, the assignment mechanism is defined by
\begin{align}
\label{eq:ass-mechanism}
&\p(\mathbf{W}=\mathbf{w}\,|\,\mathbf{Y}(0),\mathbf{Y}(1),\X,\betaall,\epsall,\thetaall) = 
\begin{pmatrix}
N \\ N_t
\end{pmatrix}^{-1}
\end{align}
for all $\W$ such that $\sum_{i=1}^N W_i = N_t$, and $0$ otherwise. Any other ignorable assignment mechanism (i.e, $p(\W\,|\,\X,\Yobs,\Ymis)=P(\W\,|\,\X,\Yobs)$ \citep{Rubin:1978} could also be assumed.



\section{Estimating the average treatment effect}
\label{sec:the-four-steps}

We estimate the ATE by explicitly imputing the missing potential outcomes in a repeated fashion to account for the uncertainty in the imputation. Let $\ymishat$ be the imputed value corresponding to $\ymis$, then the ATE in~(\ref{eq:taufs}) can be estimated as
\begin{align}
\label{eq:atehat}
   \widehat{\textrm{ATE}}=\frac{1}{N}\sum_{i=1}^N\left((2W_i-1)(\yobs-\ymishat)\right).
\end{align}


We work within the model-based Bayesian causal inference framework  \citep{Rubin:1975,Rubin:1978} and evaluate a posterior predictive distribution for at least half of the missing potential outcomes. The fundamental idea is to initiate an imputation model for the missing potential outcomes $\Ymis$, conditionally on the observed outcomes $\Yobs$ and the observed assignment vector $\W$. As we outline, the predictive posterior of ${\bf Y}^{\textrm{mis}}$ can be computed using the HMC algorithm in a package such as \texttt{rstanarm} \citep{GelmanHill:2007,Goodrich:2020}. However, this approach can be computationally expensive for large $N$. To remedy this problem, we develop an approximation algorithm in this section, which is an order of magnitude faster and gives a good approximation for moderately large $Y_i^\text{obs}$.


We first present some necessary tools required to evaluate the posterior distribution of the parameters governing the distribution of the potential outcomes. This step is computationally intractable, so instead of computing the posterior distribution exactly, we propose an approximation. In Section \ref{sec:three-one}, we present results by \cite{Bartlett-Kendall:1946} and \cite{el1973bayesian}, and extend to these Lemma \ref{lemma:convergence1} and Corollary \ref{cor:TV}, which we will use in the next subsection a new expression for the asymptotic total variation distance between our approximation and the true posterior. In Section \ref{sec:three-two}, we derive the posterior distribution of the causal estimand in four steps and in Section \ref{sec:three-three}, we compute the ATE for Poisson and lognormal-Poisson potential outcomes.

\subsection{Convergence step to derive the imputation model and divergence considerations}
\label{sec:three-one}


The main purpose of this section is to derive Corollary 1, which will be useful in the derivation of the predictive posterior of $\Ymis$. First, we note that if ${X}$ is a Gamma random variable with density function $\mathsf{Ga}(r,s)$
\begin{align*}
\frac{1}{\Gamma(r)s^r}x^{r-1}\exp\left({-\frac{x}{s}}\right),\quad x> 0,
\end{align*}
then a transformation of ${Y}:=\log {X}$ yields the log-Gamma distribution which has density
\begin{align*}
\frac{1}{\Gamma(r)s^r}\exp\left({ry-\frac{e^y}{s}}\right),\quad y>0.
\end{align*}
Denoting the log-Gamma distribution as $Y\sim\mathsf{logGamma}(r,s)$, it is well known that ${Y}$ is \emph{approximately} Gaussian distributed with mean $\log rs$ and variance $\log r^{-1}$ for large $r$ \citep{Bartlett-Kendall:1946}.
Leveraging on this result, \cite{el1973bayesian} noted that the variables and parameters of the log-Gamma and Poisson distributions are related in the following sense:
\begin{align*}
f_\text{\rm Poisson}(y\,|\,\mu\equiv e^{\xi})=\frac{e^{-\mu}\mu^y}{y!}=\frac{1}{y}\cdot\frac{e^{-e^{\xi}}(e^{\xi})^y}{\Gamma(y)}&=\frac{1}{y}\cdot f_\text{\rm logGamma}(\xi\,|\,y,1)\nonumber\\
&\overset{(\ast)}{\approx} \frac{1}{y}\cdot f_\text{\rm Normal} (\xi\,|\,\log y,y^{-1}).
\end{align*}

We now derive the total variation distance between the log-Gamma and Poisson distributions as $y\to\infty$.
The Kullback-Leibler divergence between ${f}$ and ${g}$, denoted by $\mathbb{D}({f}\,\|\,{g})$ is given by $\int {f}(x)\log(\frac{{f}(x)}{{g}(x)})dx$. Direct calculation shows that
\begin{align}
\label{eq:kl-exact}
&\mathbb{D}\left(f_{\textrm{Normal}}(\xi\,|\,\log y,y^{-1})\, \|\ \, f_{\textrm{logGamma}}(\xi\,|\,y,1)\right)\nonumber\\
&\qquad\qquad\qquad=\log\bigg(\frac{\Gamma(y)}{\sqrt{2\pi y^{-1}}}\bigg)-y\log y +ye^{\frac{1}{2y}}-\frac{1}{2}
\end{align}
as well as $\mathbb{D}(f_{\textrm{normal}}(\xi\,|\,\log y,y^{-1})\, \|\ \, f_{\textrm{log-Gamma}})\rightarrow 0 $ as $y\rightarrow\infty$.
We present the following:
\begin{lemma}
\label{lemma:convergence1}
\emph{
It holds true that
\begin{align*}
\mathbb{D}\left(f_{\textrm{Normal}}(\xi\,|\,\log y,y^{-1})\, \|\ \, f_{\textrm{logGamma}}(\xi\,|\,y,1)\right)=\frac{5}{24y}+\mathcal{O}\left(y^{-2}\right).
\end{align*}
}
\end{lemma}
\begin{proof}
Stirling's formula gives
\begin{align}
\label{eq:stirling-orig}
\Gamma(z)= \sqrt{\frac{2\pi}{z}}\left(\frac{z}{e}\right)^z\left(1+\mathcal{O}\left(z^{-1}\right)\right).
\end{align}
Expand the the Gamma function in (\ref{eq:kl-exact}) is expanded through Stirling's series \citep{Uhler59,Arfken67:buku} for its associated formula in (\ref{eq:stirling-orig}) and using the fact that $y(e^{\frac{1}{2y}}-1)\rightarrow\nicefrac{1}{2}$ when $y\rightarrow\infty$ yields the result.
\end{proof}

An alternative measure of deviation between probability measures is the total variation (TV) distance \citep{LevinPeresWilmer2006}, defined for measures $P$ and $Q$ as
\[  \text{TV}(P\|Q) = \sup_E |P(E)-Q(E)| \]
where the supremum is over all possible events $E$. Pinsker's inequality states that $\text{TV}(P\|Q)\le \sqrt{\mathbb{D}(P\|Q)}$. Using a Taylor expansion for the square root function, Lemma~\ref{lemma:convergence1} implies:
\begin{corollary}
\label{cor:TV}
\emph{
It holds true that
\begin{align*}
\text{\rm TV}\left(f_{\textrm{Normal}}(\xi\,|\,\log y,y^{-1})\, \| \, f_{\textrm{logGamma}}(\xi\,|\,y,1)\right)=\sqrt{\frac{5}{24y}}+\mathcal{O}\left(y^{-3/2}\right).
\end{align*}
}
\end{corollary}
\begin{proof}
A Taylor expansion around zero gives $\sqrt{1+\mathcal{O}(t)}=1+\mathcal{O}(t)$ as $t\to 0$. Hence we have that
\begin{align*}
\sqrt{t+\mathcal{O}(t^2)}&=\sqrt{t}\sqrt{1+\mathcal{O}(t)}=\sqrt{t}+\mathcal{O}(t^{3/2})    
\end{align*}
as $t\to 0$. Therefore, $\sqrt{y^{-1}+\mathcal{O}(y^{-2})}=y^{-1/2}+\mathcal{O}(y^{-3/2})$ as $y\to\infty$, and the corollary immediately follows from Lemma~\ref{lemma:convergence1} using Pinsker's inequality.
\end{proof}

\subsection{Imputation algorithm}
\label{sec:three-two}
We now return to our main topic, namely estimating the finite population average treatment effect. We adapt our convergence results to the setting of the proposed count potential outcomes framework. The derivation of the posterior distribution of the average treatment effect entails the following steps:
\begin{itemize}
	\item Step $(i)$. Evaluate the conditional distribution of $\Ymis$ given $\Yobs$, $\X$, $\W$, $\betaall$, $\epsall$, and $\thetaall$.
	\item Step $(ii)$. Evaluate the conditional joint distribution for the parameters $\betaall,\epsall,\thetaall$ given $\Yobs$,  $\W$, and $\X$.
	\item Step $(iii)$ Evaluate the distribution of $\Ymis\,|\,\Yobs,\W,\X$, which is the desired imputation distribution, by marginalizing the posterior predictive distribution in step $(ii)$ over $\boldsymbol{\epsilon}$, the parameter vector $\betaall$, and the hyperparameter vector $\thetaall$. 
	\item Step $(iv)$ Finally, the estimated average treatment effect is computed.
\end{itemize}

Using (\ref{eq:identity}), the conditional distribution of $\ymis$ given $\Yobs,\W,\X,\betaall,\epsall,\thetaall$ can be expressed as:
\begin{align}
\label{eq:step1}
\ymis\,|\,\Yobs,\W,\X,\betaall,\epsall,\thetaall=\mathsf{Pois}(\mumis)
\end{align}
where $\mumis:=\exp(\ximis)$ with
\begin{align}\label{ximis}
\ximis:=\left((1-W_i)\cdot\{\xit+\log \et\} + W_i\cdot\{\xic+\log \ec\}\right).
\end{align}

We immediately see from~(\ref{eq:step1}) and~(\ref{ximis}) that the conditional distribution of $\ymis$ given $\Yobs,\W,\X,\betaall,\epsall,\thetaall$ is simply equal to the marginal distribution $\ymis$ given $\W,\X,\betaall,\epsall,\thetaall$ due to the independence of the potential outcomes. Furthermore, due to the binary nature of $W_i$, we have the following equivalent expressions:
\begin{align*}
&W_i\cdot e^{{\xic}}\ec + (1-W_i)\cdot e^{\xit}\et\\
=&\exp\left((1-W_i)\cdot\{\xit+\log \et\} + W_i\cdot\{\xic+\log \ec\}\right).
\end{align*}
In order to carry out Step $(ii)$, we first need to compute the conditional distribution of $\Yobs$ given $\X,\betaall,\epsall$, and $\thetaall$. Using the fact that $\yobs = (1-W_i)\cdot\yzero+W_i\cdot \yone$, we see that
\begin{align}
\label{eq:pyobs-given-everything}
    \yobs\,|\,\W,\X,\betaall,\epsall,\thetaall = \mathsf{Pois}(\muobs)
\end{align}
where $\muobs:=\exp(\xiobs)$ with
\begin{align*}
\xiobs=(1-W_i)\{\xic+\log \ec\} + W_i\cdot\{\xit+\log\et\}.    
\end{align*}
With $\otimes$ denoting the Kronecker product, it is readily computed that
\begin{align*}
    \xiobs = \left(
    \begin{bmatrix}
    (1-W_i) && W_i
    \end{bmatrix}\otimes \x^{\top}\right)
    \begin{pmatrix}
    \betac \\ \betat
    \end{pmatrix} + \begin{bmatrix}
    (1-W_i) && W_i
    \end{bmatrix}\begin{pmatrix}
    \log\ec \\ \log\et
    \end{pmatrix}.
    \end{align*}
Define $\tilde{W}^{\textrm{obs}}_i:=\begin{bmatrix}(1-W_i) & W_i \end{bmatrix}$, $\xitilde:=
\tilde{W}^{\textrm{obs}}_i\otimes \x^{\top}$, $\tilde{m}^{\textrm{obs}}_i:=\tilde{W}^{\textrm{obs}}_i\tilde{\epsilon}_i$, and $\epstilde_i:=(\log \ec \,\, \log \et)^{\top}$, we can re-express $\xiobs$ succinctly as follows 
\begin{align}
\label{eq:xiobs-neat}
\xiobs=\xitilde\betaall + m^{\textrm{obs}}_i.
\end{align}


Due to the choice of our count potential outcomes paradigm, an analytical expression for the posterior for $\betaall$ is not available owing to the lack of conjugacy between the Gaussian priors and the Poisson likelihood. However, we now give a result that applies to the approximate $\pi(\betaall\,|\,\Yobs,\W,\X,\epsall,\thetaall)$, the posterior distribution of $\betaall$ conditional on $\Yobs$, its assignment mechanism $\W$, $\X$, $\epsall$ and its hyperparameters $\thetaall$. Some additional notation is warranted. Let $r_i:=\log \yobs-m^{\textrm{obs}}_i$,
\begin{align}
    \Xtilde^{\textrm{obs}}:=\begin{pmatrix}
    \xtilde^{\textrm{obs}}_1 \\ \xtilde^{\textrm{obs}}_2 \\ \vdots \\ \xtilde^{\textrm{obs}}_n
    \end{pmatrix},\,\,\,\Rtilde^{\textrm{obs}}:=\begin{pmatrix}
    r_1 \\ r_2 \\ \vdots \\ r_n
    \end{pmatrix},\,\,\, \text{and } \boldsymbol{\Sigma}_y^{{\textrm{obs}}}:=\begin{bmatrix}
    {(Y^{\mathrm{obs}}_1)^{-1}} & & \\
    & \ddots & \\
    & & {(Y^{\mathrm{obs}}_n)^{-1}}
  \end{bmatrix}.
\end{align}

\begin{lemma}
\label{lem:law of betaall}
\emph{
It holds true that
\begin{align*}
TV\left(p(\betaall\,|\,\Yobs,\W,\X,\epsall,\thetaall)\,,\,\mathsf{N}(\mu^{\boldsymbol{\epsilon}}_{\betaall},\Sigma_{\betaall})\right) = \mathcal{O}\Big(\sum_{i=1}^N y_i^{-3/2}\Big)
\end{align*}
where
\begin{align*}
    \mu^{\boldsymbol{\epsilon}}_{\betaall}&=\Sigma_{\betaall}(\Xtilde^{{\textrm{obs}}})^{\top}\,(\boldsymbol{\Sigma}_y^{{\textrm{obs}}})^{-1}(\tilde{{\textbf{Y}}}-\textbf{M}^{\textrm{obs}}_{\boldsymbol{\epsilon}}),\\
    \Sigma^{-1}_{\betaall}&=(\Xtilde^{{\textrm{obs}}})^{\top} (\boldsymbol{\Sigma}_y^{{\textrm{obs}}})^{-1}\Xtilde+ {(\varbeta)}^{-1}\mathbf{I}_{2(k+1)}
\end{align*}
with $\boldsymbol{\Sigma}_y^{{\textrm{obs}}}=\mathrm{diag}\left[(Y^{\mathrm{obs}}_1)^{-1},\hdots,(Y^{\mathrm{obs}}_N)^{-1}\right]$.
}
\end{lemma}

\begin{proof}
From Corollary~\ref{cor:TV}, there exists a function $\delta_i$ with the following
\begin{align*}
 \int\delta_i(\xi)d\xi=\mathcal{O}(y_i^{-1/2}) 
\end{align*}
such that
\[   f_{\textrm{logGamma}}(\xi_i\,|\,y_i,1) = f_{\textrm{Normal}}(\xi_i\,|\,\log y_i,y_i^{-1}) + \delta_i(\xi_i).   \]
Hence,
\begin{align}
    \pi(\betaall\,|\,\Yobs,\W,\X,\epsall,\thetaall)&\propto 
    p(\Yobs\,|\,\W,\X,\betaall,\epsall,\thetaall)\pi(\betaall)\nonumber\\
    &\overset{\ast}{=}\prod_{i=1}^n [f_{\mathsf{Normal}}(\xiobs\,|\,\log y_i,y^{-1}_i)+ \delta_i(\xiobs)]\pi(\betaall)\cdot\frac{1}{y_i}\nonumber\\
    & =f_{\textsf{Normal}}(\mu^{\boldsymbol{\epsilon}}_{\betaall},\Sigma_{\betaall})) +\mathcal{O}\Big(\sum_{i=1}^N y_i^{-3/2}\Big)\nonumber
\end{align}
where ${\ast}$ follows from Corollary~\ref{cor:TV} and $\delta_i'$ is a quantity such that $\int\delta_i'(\xi)d\xi=\mathcal{O}(y_i^{-1/2})$. Hence, the TV distance between the posterior $\pi(\beta|\,\cdot)$ and the $f_{\mathsf{normal}}(\cdot)$ term is of order $\mathcal{O}\Big(\sum_{i=1}^N y_i^{-3/2}\Big)$. \ok


\end{proof}

We turn our attention to investigate the conditional distribution of $\mumis$ given $\Yobs,\W,\X,\boldsymbol{\epsilon},\thetaall$.  This is important to aid the computations for Step $(iv)$. Let
$\xitildemis:=\begin{bmatrix}
W_i & (1-W_i) 
\end{bmatrix}\otimes\x^{\top}$ and $\tilde{m}^{\textrm{mis}}_i:=\begin{bmatrix}
W_i & (1-W_i) 
\end{bmatrix}\tilde{\epsilon}_i$. Then, we may write $\ximis=\xitildemis\betaall + \tilde{m}^{\textrm{mis}}_i$ and establish that 
\begin{align*}
    \ximis\,|\,\Yobs,\W,\X,\boldsymbol{\epsilon},\thetaall\sim\mathsf{N}(\xitildemis\mu^{\boldsymbol{\epsilon}}_{\betaall}+\tilde{m}^{\textrm{mis}}_i,{(\xitildemis)}^{\top}\Sigma_{\betaall}\xitildemis)
\end{align*}
by Lemma \ref{lem:law of betaall}. Hence $\ximis\,|\,\Yobs,\W,\X,\boldsymbol{\epsilon},\thetaall$ is approximately log-Gamma from Lemma \ref{lemma:convergence1} with parameters $((\xitildemis)^{\top}\Sigma_{\betaall}\xitildemis)^{-1}$ and $(\xitildemis)^{\top}\Sigma_{\betaall}\xitildemis\,\exp(\xitildemis\mu^{\boldsymbol{\epsilon}}_{\betaall}+\tilde{m}^{\mathrm{mis}}_i)$. Since $\mumis=\exp(\ximis)$, we see that the conditional distribution of $\mumis$ given $\Yobs,\W,\X,\boldsymbol{\epsilon},\thetaall$ is approximately Gamma distributed ($\mathsf{Ga}(\cdot,\cdot)$), i.e.,
\begin{align}
\label{eq:mumis-distribution}
\mumis\,|\,\Yobs,\W,\X,\boldsymbol{\epsilon},\thetaall\sim\mathsf{Ga}(\gamma_i,h_i)
\end{align}
where
\begin{align}
\label{eq:mumis-distribution-2}
\gamma_i:=((\xitildemis)^{\top}\Sigma_{\betaall}\xitildemis)^{-1}, \quad  h_i:=(\xitildemis)^{\top}\Sigma_{\betaall}\xitildemis\,\exp(\xitildemis\mu^{\boldsymbol{\epsilon}}_{\betaall}+\tilde{m}^{\textrm{mis}}_i).
\end{align}

\fbox{\begin{minipage}{33em}
\label{procedure1}
\emph{Procedure 1}: Sampling (${\betaall}, \epsall, \thetaall$).
\\
{Input}: $N$ units of training samples of the form ($\Yobs, \textbf{x}_i, W_i$) where $\textbf{x}_i$ are covariates, $\Yobs$ is the observed outcome, and $W_i$ is the treatment assignment
\begin{enumerate}[noitemsep,topsep=0pt,label*=\arabic*.]
    \item Initialize a list $\mathcal{S}$ of samples as empty
    \item Randomly initialize ($\betaall^0, \epsall^0, \thetaall^0$) and append to $\mathcal{S}$
    \item Repeat the following steps until the maximum number of iteration is reached
    \begin{enumerate}[noitemsep,topsep=0pt,label*=\arabic*.]
        \item Draw a random sample of $\betaall$ using Lemma \ref{lem:law of betaall} conditioned on the latest sample of $\epsall$ and $\thetaall$ in the list $\mathcal{S}$.
        \item Draw a random sample of $\epsall$ using the density in equation (\ref{eq:marginal_posterior}) conditioned on the latest sample of $\betaall$ and $\thetaall$ in the list $\mathcal{S}$.
        \item Draw a random sample of $\thetaall$ from the density in equation (\ref{eq:marginal_posterior})  conditioned on the latest sample of $\betaall$ and $\epsall$ in the list $\mathcal{S}$.
    \end{enumerate}
    \item Discard the first $m$ burn-in samples in $\mathcal{S}$ so that $|\mathcal{S}|=R$ which we call $\mathcal{S}^R$
    \item Return the samples of the joint posterior $\mathcal{S}^R$
\end{enumerate}
\end{minipage}}

\vspace{2mm}

We now outline an MCMC scheme to sample from the posterior predictive of $\Ymis$ conditional on $\Yobs$, $\W$ and its covariates. We can express the following
\begin{equation}
\label{eq:step3-useful}
\begin{aligned}
&\p(\ymis=y\,|\,\Yobs,\W,\X)\\
&=\int\int\int {p(Y_i^{\mathrm{mis}}=y\,|\,\Yobs,\W,\X,\epsall,\betaall,\thetaall)}\cdot {p(\epsall,\betaall,\thetaall\,|\,\Yobs,\W,\X)}d\epsall\,d\betaall\,d\thetaall.
\end{aligned}
\end{equation}
The quantities $\betaall,\epsall$ and $\thetaall$ are sampled from the joint posteriors, which is denoted by  
$p(\epsall,\betaall,\thetaall\,|\,\Yobs,\W,\X)$ using a Gibbs sampler from their marginal posteriors with the following densities
\begin{align}
\label{eq:marginal_posterior}
    p(\betaall\,|\,\epsall,\thetaall,\Yobs,\W,\X),\,\,\,\,\, p(\epsall\,|\,\betaall,\thetaall,\Yobs,\W,\X)\,\,\,\,\text{and}\,\,\,\,
    p(\thetaall\,|\,\epsall)
\end{align}
where $p(\betaall\,|\,\epsall,\thetaall,\Yobs,\W,\X)$ is Gaussian from Lemma \ref{lem:law of betaall} and the posterior distribution for the hyperparameters $p(\thetaall\,|\,\epsall)$ is known. We further remark that $p(\thetaall\,|\,\epsall)$ does not depend on $\betaall$ but depends implicitly on $\Yobs,\W$ and $\X$ through $\epsall$. Depending on the choice of $\epsall$, $p(\epsall\,|\,\betaall,\thetaall,\Yobs,\W,\X)$ may be tractable, as we shall see from the case of lognormal-Poisson potential outcomes in the next subsection. A summary of the steps needed to sample $p(Y_i^{\mathrm{mis}}=y\,|\,\Yobs,\W,\X,\epsall,\betaall,\thetaall)$ is given in Procedure 1. Once $\betaall,\epsall$ and $\thetaall$ are known, use these values to further sample ${p(Y_i^{\mathrm{mis}}=y\,|\,\Yobs,\W,\X,\epsall,\betaall,\thetaall)}$ from a Poisson law as given in (\ref{eq:step1}). 

\textbf{The special case of $\epsall=1$.} Next we turn to a discussion when $\epsall=1$, i.e., the case of Poisson potential outcomes. We show that an MCMC sampler is not needed in this special case since an analytical form for $\p(\ymis=y\,|\,\Yobs,\W,\X,\epsall,\thetaall)$ exists. In the interest of generality, an expression that includes $\epsall$ and $\thetaall$ is presented as it offers insights through the potential outcomes paradigm as to why this formulation only works for $\epsall=1$ and not for general $\epsall$ and $\thetaall$.

Note that the conditional density of $\ymis=y_i$ given $\Yobs,\W$ can be expressed as:
\begin{align*}
\p(\ymis=y\,|\,\Yobs,\W,\X)=\int\int \psi_{\epsall,\thetaall}(y)\p(\epsall,\thetaall\,|\Yobs,\W,\X)   d\epsilon_i \,d\vartheta_i
\end{align*}
where 
\begin{align}
\label{eq:inner-integral}
\psi_{\epsall,\thetaall}(y)=\begin{pmatrix}
\gamma+y-1 \\ y
\end{pmatrix}\left(\frac{1}{1+h}\right)^{\gamma}\left(1-\frac{1}{1+h}\right)^{y}
\end{align}
with $\gamma$ and $h$ given in equation (\ref{eq:mumis-distribution-2}). To see why this is true, first observe that 
\begin{align}
\label{eq:step3-useless}
&\p(\ymis=y\,|\,\Yobs,\W,\X)\nonumber\\
&=\int\int\int {p(Y^{\mathrm{mis}}=y\,|\,\Yobs,\W,\X,\epsall,\mu^{\textrm{mis}},\thetaall)}\cdot \p(\mumis\,|\,\epsall,\thetaall,\Yobs,\W,\X)\cdot\nonumber\\
&\qquad\qquad \p(\epsall,\thetaall\,|\,\Yobs,\W) d\epsall \,d\mumis\,d\thetaall
\end{align}
and define the function $\psi$ as the inner integral with respect to $\mumis$:
\begin{align}
\label{eq:psi}
&\psi_{\epsall,\thetaall}(y)\nonumber\\
&:=\int \p(\ymis=y\,|\,\mumis,\epsall,\thetaall,\Yobs,\W,\X)\cdot p(\mumis\,|\,\epsall,\thetaall,\Yobs,\W,\X)d\mumis.
\end{align}
Then, note that the first and second terms of the integrand for $\psi_{\epsall,\thetaall}(y)$ are Poisson and Gamma, respectively. Hence, $\psi_{\epsall,\thetaall}(y)$ is the (conditional) density of a negative-binomial distribution as a result of the Poisson-Gamma mixture property which yields equation (\ref{eq:inner-integral}). Equations (\ref{eq:step3-useful}) and (\ref{eq:step3-useless}) differ in that the former has $\betaall$ marginalized while the latter has $\mu^{\textrm{mis}}$ integrated out. The latter expression is difficult to sample from due to the construction of our count potential outcomes. To elucidate on this point, observe that $p(\epsall,\thetaall\,|\,\Yobs,\W,\X)$ in equation (\ref{eq:step3-useless}) can be decomposed into the product of two components
\begin{align}
\label{eq:joint-posterior-useless}
    p(\epsall,\thetaall\,|\Yobs,\W,\X)=p(\epsall\,|\Yobs,\W,\X,\thetaall)\cdot p(\thetaall\,|\,\Yobs,\W,\X).
\end{align}
The first and the second terms on the right hand side of equation (\ref{eq:joint-posterior-useless}) can be re-expressed as
\begin{align}
\label{eq:tmp1}
    p(\epsall\,|\Yobs,\W,\X,\thetaall)\propto p(\Yobs\,|\,\W,\X,\thetaall,\epsall)\cdot \p(\epsall\,|\,\thetaall)
\end{align}
as well as 
\begin{align}
\label{eq:tmp2}
    \p(\thetaall\,|\,\Yobs,\W,\X)\propto p(\Yobs\,|\,\W,\X,\thetaall)\cdot\p(\thetaall),
\end{align}
respectively. Since we can determine the distribution of the overdispersion $\p(\epsall\,|\,\thetaall)$ and its hyperparameters $\p(\thetaall)$, this makes it easier to sample from. However, we do not know, \emph{a priori}, the conditional distributions of $p(\Yobs\,|\,\W,\X,\thetaall,\epsall)$ and $p(\Yobs\,|\,\W,\thetaall,\epsall)$ from the first terms of equations (\ref{eq:tmp1}) and (\ref{eq:tmp2}), respectively. But what we \emph{do} have, is the posterior predictive distribution of $\Yobs$, i.e., the conditional distribution of the missing potential outcomes given the $\W,\X,\epsall,\thetaall$, whose explicit form is given in equation (\ref{eq:pyobs-given-everything}). Even though it is evident from the preceding discussion that the decomposition in equation (\ref{eq:step3-useless}) is prohibitive in drawing samples for general quantities of $\betaall$ and $\epsall$, it is particularly useful to evaluate the conditional distribution of $\Ymis$ given $\Yobs,\W$ and the covariates for a Poisson potential outcomes model where an explicit expression is available. Hence, the conditional density of $\ymis=y_i$ given $\Yobs,\W,\X$ can be explicitly expressed as:
\begin{align}
\label{eq:step1-for-poisson}
\p(\ymis=y\,|\,\Yobs,\W,\X)= \begin{pmatrix}
\gamma+y-1 \\ y
\end{pmatrix}\left(\frac{1}{1+h}\right)^{\gamma}\left(1-\frac{1}{1+h}\right)^{y}
\end{align}
whose values of $\gamma$ and $h$ are given in equation (\ref{eq:mumis-distribution-2}).

We complete this section by carrying out Step $(iv)$ to compute the ATE as in (\ref{eq:taufs}). As discussed previously, for general specifications of $\epsall$ and $\thetaall$, the distribution of ATE may not have an analytical form. Here, we explain the sampling methods for evaluating it. The key elements are Steps $(i)$ and $(ii)$. Given a draw of $\thetaall$, $\epsall$ and $\betaall$ (drawn using Procedure 1), we substitute these values into the conditional distribution of $\Ymis$ using Step $(i)$ to impute, independently all the missing potential outcomes, that is, draw $\Ymis$ from a Poisson distribution by substituting the pairs of $\thetaall$, $\epsall$, and $\betaall$ as in Step $(i)$ using equation (\ref{eq:step1}). Substituting the observed and imputed missing potential outcomes into (\ref{eq:atehat}) yields an estimate for ATE based on the $r^{th}$ imputation, $\hat\tau^{(r)}$
\begin{align}
\label{eq:tauhatfs}
    \widehat{\textrm{ATE}}^{(r)}=\frac{1}{N}\sum_{i=1}^N\left((2W_i-1)\yobs+(1-2W_i)\hat Y^{{(r)}}_i\right)
\end{align}
where $\hat Y^{{(r)}}_i$
where we have used the identity in equation (\ref{eq:identity}). To derive the full posterior distribution of ATE, this procedure is repeated $R$ times and the average and variance of the imputed estimators $\widehat{\textrm{ATE}}^{(1)},\widehat{\textrm{ATE}}^{(2)},\hdots,\widehat{\textrm{ATE}}^{(R)}$ are
\begin{align}
\label{eq:estimated-mean-variance-tau}
    \frac{1}{{R}}\sum_{r=1}^{R}\widehat{\textrm{ATE}}^{(r)}=:\overline{\textrm{ATE}},\qquad \frac{1}{{R-1}}\sum_{r=1}^{R} (\widehat{\textrm{ATE}}^{(r)}-\overline{\textrm{ATE}})^2.
\end{align}

\fbox{\begin{minipage}{35em}
\emph{Procedure 2}: Estimating the mean and variance of $\hat{\tau}$.
\\
{Input}: A list $\mathcal{S}$ containing samples of (${\betaall}, \epsall, \thetaall$) from the joint posterior, $N$ training examples of the form $(\Yobs, \mathbf{x}_i, W_i)$ where $\mathbf{x}_i$ are covariates, $\Yobs$ is the observed outcome, and $W_i$ is the treatment assignment

\begin{enumerate}[noitemsep,topsep=0pt,label*=\arabic*.]
\item Initialize an empty list $\mathcal{T}^R$ of length $R$
\item For each $r$-th triplet (${\betaall},\epsall, \thetaall$) in the list $\mathcal{S}$ returned from Procedure~1
\begin{enumerate}[noitemsep,topsep=0pt,label*=\arabic*.]
    \item Initialize an empty list $\mathcal{L}^N_{\Ymis}$ to store $\Ymis$ of length $N$
    \item For each individual $i \in \{1,2,\dots,N\}$
    \begin{enumerate}[noitemsep,topsep=0pt,label*=\arabic*.]
        \item Draw a sample of $Y_i^{\textrm{mis}}\,|\,\Yobs,\W,\X,\betaall,\epsall,\thetaall$ using equation (\ref{eq:step1})
        \item Append the above $Y_i^{\textrm{mis}}$ sample to $\mathcal{L}^N_{\Ymis}$
    \end{enumerate}
    \item Compute $\widehat{\textrm{ATE}}^{(r)}$ using equation (\ref{eq:tauhatfs}) and $\mathcal{L}^N_{\Ymis}$
    \item Append the above $\widehat{\textrm{ATE}}^{(r)}$ to $\mathcal{T}^R$
\end{enumerate}
\item Compute empirical mean and variance of $\widehat{\textrm{ATE}}$ from list $\mathcal{T}^R$ using equation (\ref{eq:estimated-mean-variance-tau})
\item Return the empirical mean and variance of $\widehat{\textrm{ATE}}$.
\end{enumerate}
\end{minipage}}

Procedure 2 summarizes the preceding discussion on sampling methods to compute average treatment effects in our paradigm.

\subsection{Case studies with Poisson and lognormal-Poisson potential outcomes}
\label{sec:three-three}
In the previous section, we developed a Bayesian imputation model for the missing potential outcomes with an approximation architecture to bypass the non conjugacy between the Gaussian priors and the Poisson likelihood. The result is a general strategy for drawing samples for the predictive posterior for $\Ymis$ which is then used to compute the finite population average treatment effect. In this section, we show via two examples, one of which exhibits a closed form expression for the ATE (Poisson potential outcomes), and another, for which the ATE can be found with an efficient simulation method due to tractable marginal posteriors (lognormal-Poisson potential outcomes).

\paragraph{Poisson potential outcomes.} We set $\epsall\equiv 1.$ Since there are no hyperparameters, the resulting conditional distribution of $\ymis$ given $\Yobs,\W$ and its covariates $\X$ is given by equation (\ref{eq:step1-for-poisson}). The exact distribution of the finite population average treatment effect given in terms of its conditional mean and variance are as follows:
\begin{align}
\label{eq:poissonate}
\frac{1}{N}\sum_{i=1}^N\left((2W_i-1)\cdot (\yobs-\gamma_i)\right),\quad \frac{1}{N^2}\sum_{i=1}^N \gamma_i\frac{(1-h_i)}{h_i^2}
\end{align}
where $\gamma_i$ and $h_i$ are given in equation (\ref{eq:mumis-distribution-2}).

\paragraph{Lognormal-Poisson potential outcomes.} In this example, we take $\epsilon^{[c]}\sim\log N(0,(\sigma^{[c]})^2)$ and 
$\epsilon^{[t]}\sim\log N(0,(\sigma^{[t]})^2)$. The hyperparameters for this model are $\thetaall=(\sigma^{[c]},\sigma^{[t]})$ with the following priors
\begin{align*}
    \sigma^{[c]}\sim\mathsf{IG}(\alpha^{[c]},\nu^{[c]})\,\,\text{and}\,\,\sigma^{[t]}\sim\mathsf{IG}(\alpha^{[t]},\nu^{[t]})
\end{align*}
as in equation (\ref{eq:priors_overdispersion}), as described in Section \ref{sec:count-potential-outcomes} (repeated here for convenience).

We now spell out the marginal posteriors in (\ref{eq:marginal_posterior}). It is readily computed that the posterior of $\epsall$ is given by
\begin{equation}
	\begin{aligned}
	\p(\epsilon_i^{[c]} \,|\, \Yobs, \W,\X,\betaall,\thetaall) &= \log\mathsf{N}(\epsilon_i^{[c]}\,|\,m_i^{[c]}, (s_i^{[c]})^2),\\
	(s_i^{[c]})^2&= ((1-W_i)^2 \yobs + (\sigma^{[c]})^{-2})^{-1},\\
	m_i^{[c]} &= -(1-W_i)\yobs (\xitilde\betaall-\log \yobs)(s_i^{[c]})^2
	\end{aligned}
	\end{equation}
	and
	\begin{align}
	\p(\epsilon_i^{[t]} \,|\,\Yobs, \W,\betaall,\thetaall) &= \log\mathsf{N}(\epsilon_i^{[t]}\,|\,m_i^{[t]}, (s^{t})^2),\nonumber\\
	(s_i^{[t]})^2 &= (W_i^2\yobs + (\sigma^{[t]})^{-2})^{-1},\\
	m_i^{[t]} &= -W_i \yobs (\xitilde\betaall-\log \yobs)(s_i^{[t]})^2\nonumber
\end{align}
where $\xitilde\betaall$ is defined in (\ref{eq:xiobs-neat}). Finally, the posterior of the hyperparameters are given by
\begin{equation}
\label{eq:hyper1}
\begin{aligned}
\p\left((\sigma^{[c]})^2 \,|\,\epsilon^{[c]}\right) &= \mathsf{IG}\left((\sigma^{[c]})^2|\tilde\alpha^{[c]}, \tilde\gamma^{[c]}\right),\\
\tilde\alpha^{[c]} = \alpha^{[c]} + \frac{1}{2}N,&\quad
\tilde\gamma^{[c]}= \gamma^{[c]} + \frac{1}{2}\sum_{i=1}^N (\log \epsilon_i^{[c]})^{2}
\end{aligned}
\end{equation}
and
\begin{equation}
\label{eq:hyper2}
\begin{aligned}
\p((\sigma^{[t]})^2|\epsilon^{[t]}) &= \mathsf{IG}((\sigma^{[t]})^2|\tilde\alpha^{[t]}, \tilde\gamma^{[t]}),\\
\tilde\alpha^{[t]} = \alpha^{[t]} + \frac{1}{2}N,&\quad
\tilde\gamma^{[t]}= \gamma^{[t]} + \frac{1}{2}\sum_{i=1}^N (\log \epsilon_i^{[t]})^{2}.
\end{aligned}
\end{equation}
The closed form for the marginal posterior for $\betaall$ is guaranteed by Lemma \ref{lem:law of betaall}. Moreover, the tractability of the marginal posteriors of the overdispersion $\epsall$ and its hyperparamters $\thetaall$ makes it appealing in this setting as the average treatment effect can be computed efficiently using the steps outlined in Procedure 2.


\section{Simulations and real data analysis}
\label{sec:experiments}
We first assess the newly proposed approximation in a controlled setting. Specifically, we illustrate how the sample size, the types of potential outcomes, and the variation of the model complexities affect the error of our approximation methods. We compare the ATE estimated from an exact method with no approximation ($\mathsf{HMC\mbox{-}Exact}$) versus that of using our proposed approximation method ($\mathsf{Our\mbox{-}Approx}$). Working within the potential outcomes model in (\ref{eq:model}), our proposed approximation to estimate the finite population average treatment effect uses Lemma \ref{lem:law of betaall} as well as Procedures 1 and 2. The exact method, however, does not utilize any such approximations. In this case, we draw exact samples using an MCMC scheme for the posterior distributions using \texttt{rstanarm} \citep{GelmanHill:2007,Goodrich:2020}. We consider two types of potential outcomes: potential outcomes that do not (i.e., Poisson, $\epsilon\equiv 1$) and do (i.e., lognormal-Poisson, $\epsilon^{[c]}=\epsilon^{[t]}\sim\log N(0,\sigma^2))$) account for overdispersion.

\subsection{Parameter recovery and computational time}
\label{sec:synthetic}
    \begin{figure}[!ht]
        \centering
        \begin{subfigure}[b]{0.475\textwidth}
            \centering
            \includegraphics[width=\textwidth]{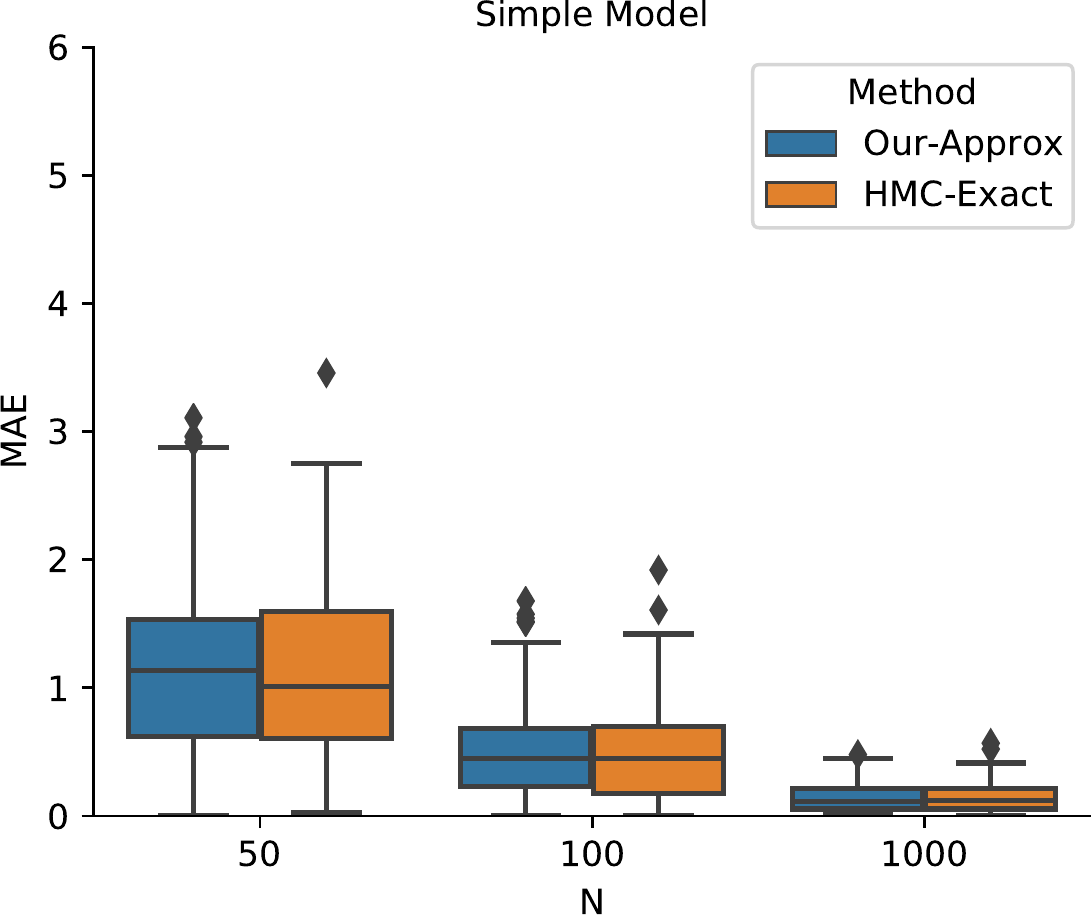}
            \caption[Network2]%
            {{\tiny Poisson potential outcomes: Simple Model}}    
            \label{fig:mean and std of net14}
        \end{subfigure}
        \hfill
        \begin{subfigure}[b]{0.475\textwidth}  
            \centering 
            \includegraphics[width=\textwidth]{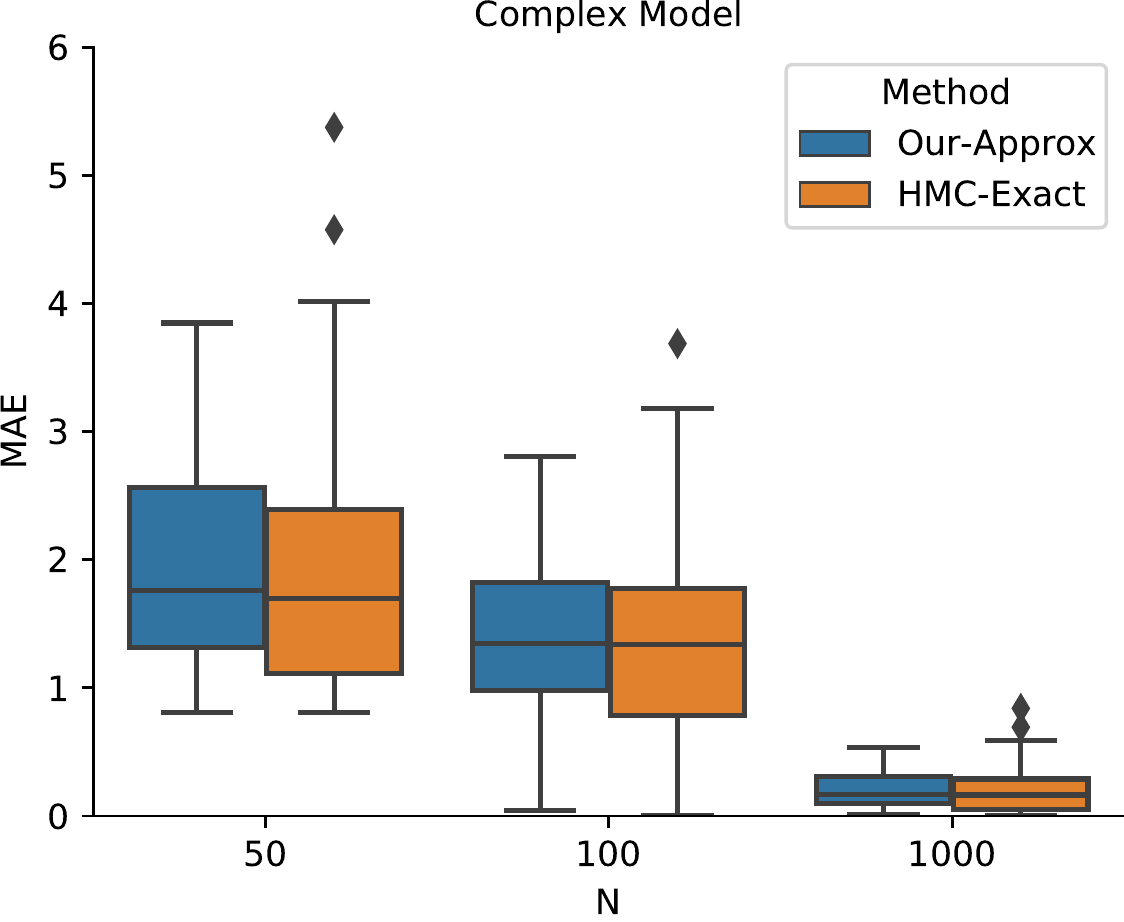}
            \caption[]%
            {{\tiny Poisson potential outcomes: Complex Model}}    
            \label{fig:mean and std of net24}
        \end{subfigure}
        \vskip\baselineskip
        \begin{subfigure}[b]{0.475\textwidth}   
            \centering 
            \includegraphics[width=\textwidth]{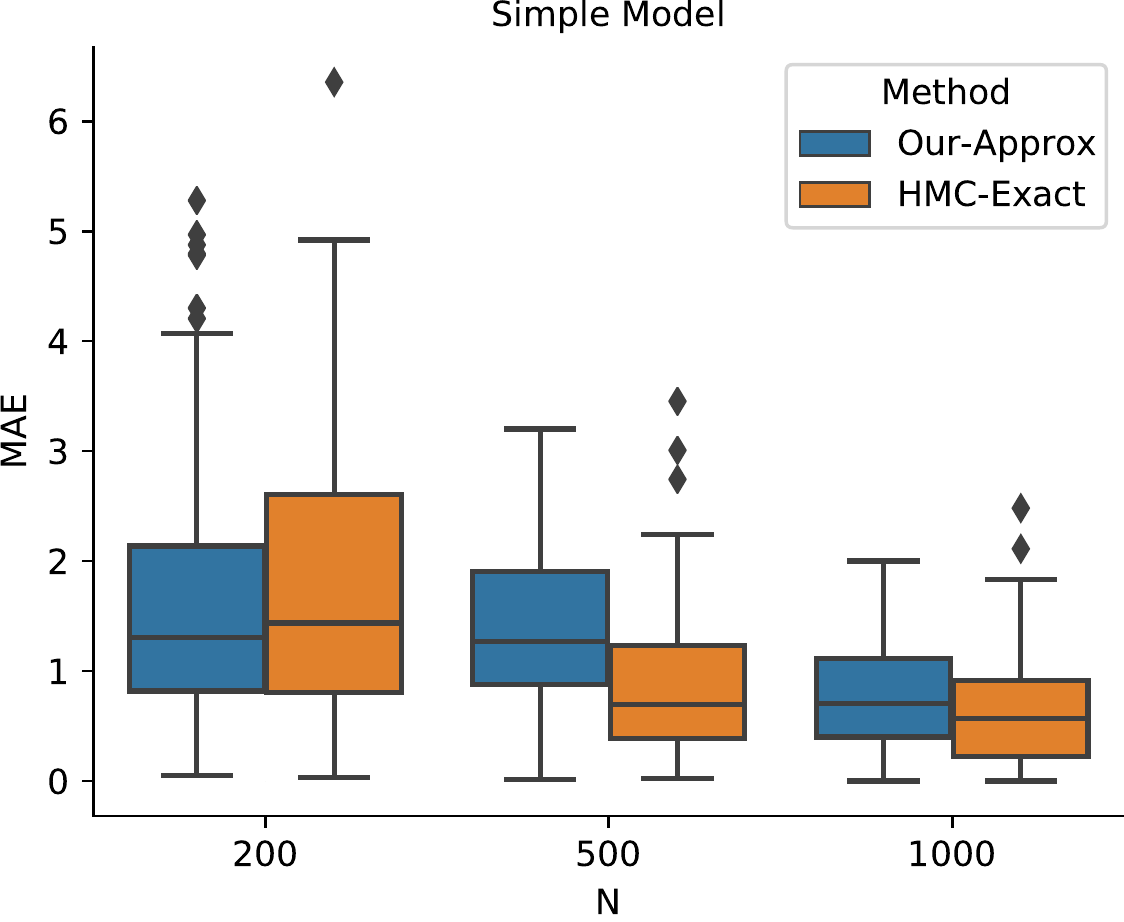}
            \caption[]%
            {{\tiny lognormal-Poisson potential outcomes: Simple Model}}    
            \label{fig:mean and std of net34}
        \end{subfigure}
        \hfill
        \begin{subfigure}[b]{0.475\textwidth}   
            \centering 
            \includegraphics[width=\textwidth]{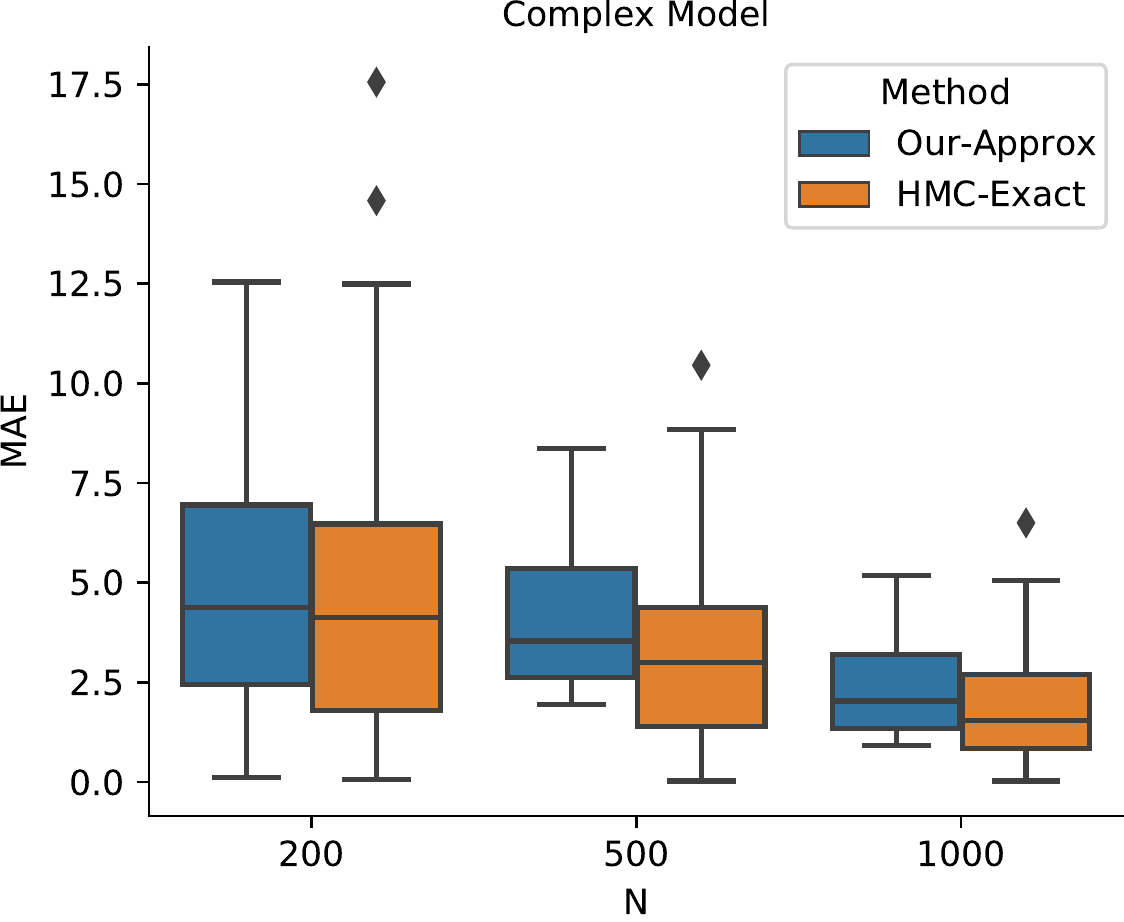}
            \caption[]%
            {{\tiny lognormal-Poisson potential outcomes: Complex Model}}    
            \label{fig:mean and std of net44}
        \end{subfigure}
        \caption[ The average and standard deviation of critical parameters ]
        {\small MAE of ATE estimation with Poisson (overdispersion parameter $\epsall=1)$) and lognormal-Poisson potential outcomes (overdispersion parameter $\epsall=\log N(0,\sigma^2)$) for the Simple and Complex Models in equation (\ref{eq:simple-complex-model}). A comparison between our approximation ($\mathsf{Our\mbox{-}Approx}$) and exact inference ($\mathsf{HMC\mbox{-}Exact}$) is presented.} 
        \label{fig:combined-boxplots}
    \end{figure}
\paragraph{\textbf{Setup.}}
We describe our simulation study in terms of the number of units $N$, where for each $N$, we compute the mean absolute error (MAE) of the ATE defined in (\ref{eq:taufs}) for both Poisson and the lognormal-Poisson potential outcomes. Precisely, we evaluate the MAE, defined as $\|{\textrm{ATE}}_{\mathsf{Our\mbox{-}Approx}}-{\textrm{ATE}}_{\ast}\|_1$ for each $N$ where $\textrm{ATE}_{\ast}$ denotes the true ATE. We further set the number of treated and control units as $N_t=N_c=\nicefrac{1}{2}\cdot N$.

In evaluating the performance of our approximation on the ATE, we generate $N$ samples of the potential outcomes model with known $N_t=N_c$ and true parameters  $\betaall_{\ast}=(\beta_{\ast}^{[t]},\beta_{\ast}^{[c]})^{\top}$ for each $\epsilon$. From these samples, we compute a parameter estimate $\hat\betaall_{\mathsf{Our\mbox{-}Approx}}$ using Procedure 1. To aid comparison, we further draw samples using an exact scheme for the posterior of $\betaall$ using \texttt{rstanarm} to obtain the distributions of $\hat\betaall_{\mathsf{HMC\mbox{-}Exact}}$. A similar set of estimates is obtained for the hyperparameters $\thetaall_{\ast}=(\vartheta^{[t]}_{\ast},\vartheta^{[c]}_{\ast})$. Given these inferred parameters $\hat\betaall_{\mathsf{HMC\mbox{-}Exact}}$ and $\hat\betaall_{\mathsf{Our\mbox{-}Approx}}$ as well as $\hat\thetaall_{\mathsf{HMC\mbox{-}Exact}}$ and $\hat\thetaall_{\mathsf{Our\mbox{-}Approx}}$, we compute the corresponding ATE for the two cases of potential outcomes. 

We consider two simulation models, which we term "Simple Model" and "Complex Model", and whose true parameters $\betaall_{\ast}$ are varied for each $\epsilon$ under consideration:
\begin{equation}
\label{eq:simple-complex-model}
\begin{aligned}
    \text{Simple~Model}\!&= \!\!
\begin{cases}
    \mu_i^{[c]}\!\!=\!\exp(3.2\!+\!0.3x_1)\epsilon^{[c]}_i\\
    \mu_i^{[t]}\!\!=\!\exp(3.7\!+\!0.8x_1)\epsilon^{[t]}_i,
\end{cases}\\
\text{Complex~Model}\!&= \!\!
\begin{cases}
    \mu_i^{[c]}\!\!=\!\exp(3.2\!+\!0.3x_{i1}\!+\!0.7x_{i2}\!+\!1.0x_{i3}\!+\!0.4x_{i4}\!+\!0.8x_{i5})\epsilon^{[c]}_i\\
    \mu_i^{[t]}\!\!=\!\exp(3.7\!+
    \!0.8x_{i1}\!+\!0.5x_{i2}\!+\!1.2x_{i3}\!+\!0.6x_{i4}\!+\!0.9x_{i5})\epsilon^{[t]}.
\end{cases}
\end{aligned}
\end{equation}

The quantities $x_{id}$ are drawn independently from a uniform distribution on the interval $[-1,1]$ for individual $i\in\{1,\hdots,N\}$, and of dimension $d\in\{1,\hdots,5\}$. With the Poisson potential outcomes, the ATE has a closed form solution (\ref{eq:poissonate}). For lognormal-Poisson, we use the closed form expressions for the posteriors for hyperparmeters in equations (\ref{eq:hyper1})--(\ref{eq:hyper2}) and repeat this $R=1000$ times as explicated in Procedures 1 and 2 to obtain the ATE.

In addition, we offer some insight into the behavior and frequency of the occurrence of the count potential outcomes had we used a continuous potential outcomes framework to model its count counterpart. To that end, using the simulated count potential outcomes of the Complex Model, we first use the identity $\yobs = (1-W_i)\cdot\yzero+W_i\cdot \yone$ and then use Bayesian Additive Regression Trees \citep{hill2011bayesian} to fit $\yobs$ against the covariates and $W_i$ to infer the parameters. With these, the potential outcomes $Y(0)$ and $Y(1)$ are simulated using the Bayesian Additive Regression Trees (confer the functions $f(W=0,x)$ and $f(W=1,x)$ in Section 3 of \cite{hill2011bayesian}).

\paragraph{Results.}

Figure~\ref{fig:combined-boxplots} shows boxplots for the variation of the estimation error (MAE) of our proposed approximation vs exact methods. We see that the estimates using our approximation have commensurate accuracy to the exact inference for various $N$: when our approximation is well specified, the MAEs obtained from the approximation have little difference from that of the exact methods, and indeed recovered the optimal ATEs asymptotically.

One incontrovertible advantage of using our proposed approximation is its speed – it is seen to be orders of magnitude faster compared to drawing inferences using exact methods. Tables \ref{tab:compute-time-poisson} and \ref{tab:compute-time-lognormal-poisson} provide the compute time to estimate the ATE for the  Complex Model. We see that our method requires less computational time (more than 100-fold) than using \texttt{rstanarm} \citep{GelmanHill:2007,Goodrich:2020} to approximate the same posterior distribution and subsequently compute the ATE.

\begin{table}
\centering
\scalebox{0.8}{
\begin{tabular}{@{}lll@{}}
\toprule
$N$     \quad    & \quad HMC-Exact                    & \quad  Our-Approx \\ \midrule
100,000 \quad  & \quad \textcolor{white}{00}17 mins                   & \quad \textcolor{white}{0}0.1 mins              \\
200,000   \quad & \quad \textcolor{white}{00}33 mins                    & \quad \textcolor{white}{0}0.2 mins             \\
500,000   \quad & \quad \textcolor{white}{00}81 mins ($\sim$1.3 hours)  & \quad \textcolor{white}{0}0.5 mins              \\
1,000,000  \quad & \quad \textcolor{white}{0}162 mins ($\sim$2.7 hours) &  \quad \textcolor{white}{0}1.0 mins              \\
2,000,000 \quad & \quad \textcolor{white}{0}333 mins ($\sim$5.5 hours) & \quad \textcolor{white}{0}2.0 mins              \\
5,000,000 \quad & \quad $>$ $24$ hours             & \quad \textcolor{white}{0}5.0 mins             \\
10,000,000 \quad & \quad $>$ $24$ hours             & \quad 10.0 mins             \\ \bottomrule
\end{tabular}}
\caption{Computational time in minutes (mins) to estimate the ATE for Complex Model with Poisson potential outcomes, see equation (\ref{eq:simple-complex-model}) with $\boldsymbol\epsilon=1$.}
\label{tab:compute-time-poisson}
\end{table}

\begin{table}
\centering
\scalebox{0.8}{
\begin{tabular}{@{}lll@{}}
\toprule
$N$ \quad       &  \quad HMC-Exact                 &  \quad Our-Approx \\ \midrule
10,000   \quad  & \quad  \textcolor{white}{00}75 mins ($\sim$1.3 hours)                    & \quad  \textcolor{white}{0}0.3 mins          \\
20,000    & \quad  \textcolor{white}{0}100 mins ($\sim$1.6 hours)                    & \quad   \textcolor{white}{0}0.5 mins              \\
50,000  \quad   & \quad  \textcolor{white}{0}260 mins ($\sim$4.3 hours)  & \quad  \textcolor{white}{0}1.3 mins             \\
80,000  \quad   &  \quad  \textcolor{white}{0}162 mins ($\sim$2.7 hours) 
& \quad  \textcolor{white}{0}2.0 mins              \\
100,000 \quad  & \quad  \textcolor{white}{0}450 mins ($\sim$7.5 hours) & \quad  \textcolor{white}{0}2.5 mins              \\
500,000 \quad   & \quad  $>$ $24$ hours             & \quad  11.0 mins       \\ 
1,000,000 \quad   & \quad  $>$ $24$ hours             & \quad  20.0 mins       \\ \bottomrule
\end{tabular}}
\caption{Computational time in minutes (mins) to estimate the ATE for Complex Model with lognormal-Poisson potential outcomes, see equation (\ref{eq:simple-complex-model}) with overdispersion parameter $\epsall=\log N(0,0.5^2)$.}
\label{tab:compute-time-lognormal-poisson}
\end{table}

\begin{figure}
\centering
\begin{subfigure}[b!]{0.96\textwidth}
   \includegraphics[width=1\linewidth]{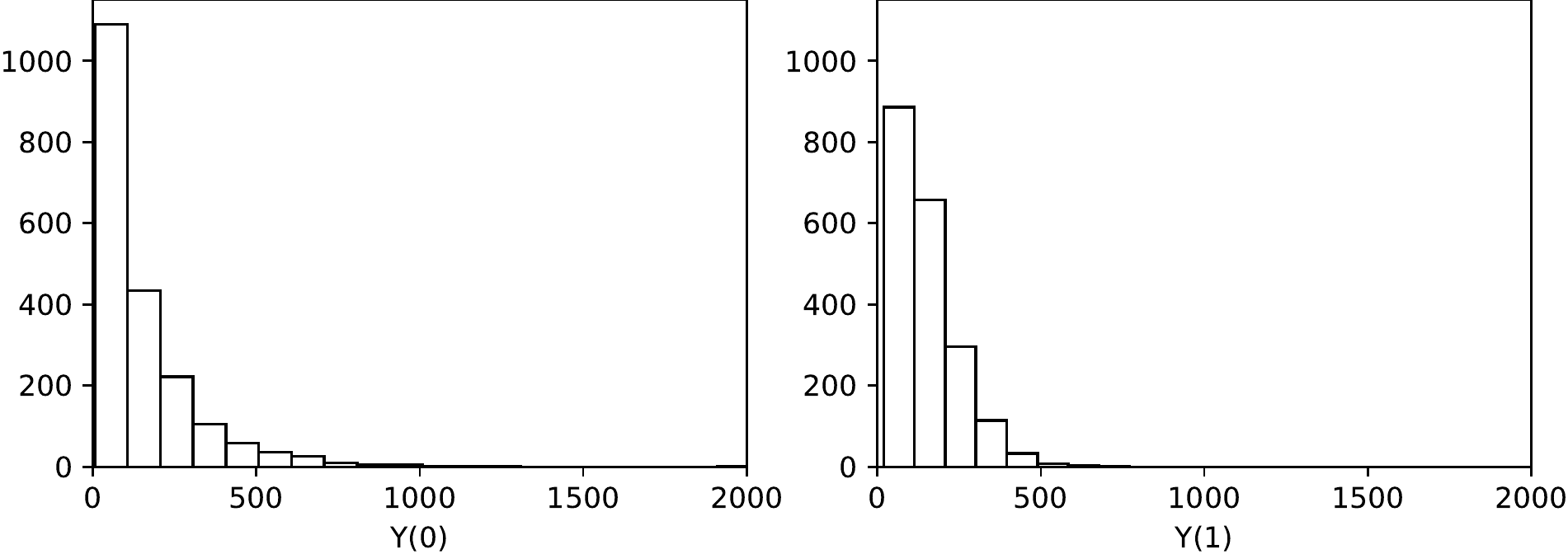}
   \caption{Histogram of potential outcomes drawn from the Complex Model in equation (\ref{eq:simple-complex-model}) with lognormal-Poisson potential outcomes (overdispersion parameter $\epsall=\log N(0,\sigma^2)$)}
   \label{fig:Ng1} 
\end{subfigure}
$\newline\newline$
\begin{subfigure}[b!]{0.94\textwidth}
   \includegraphics[width=1\linewidth]{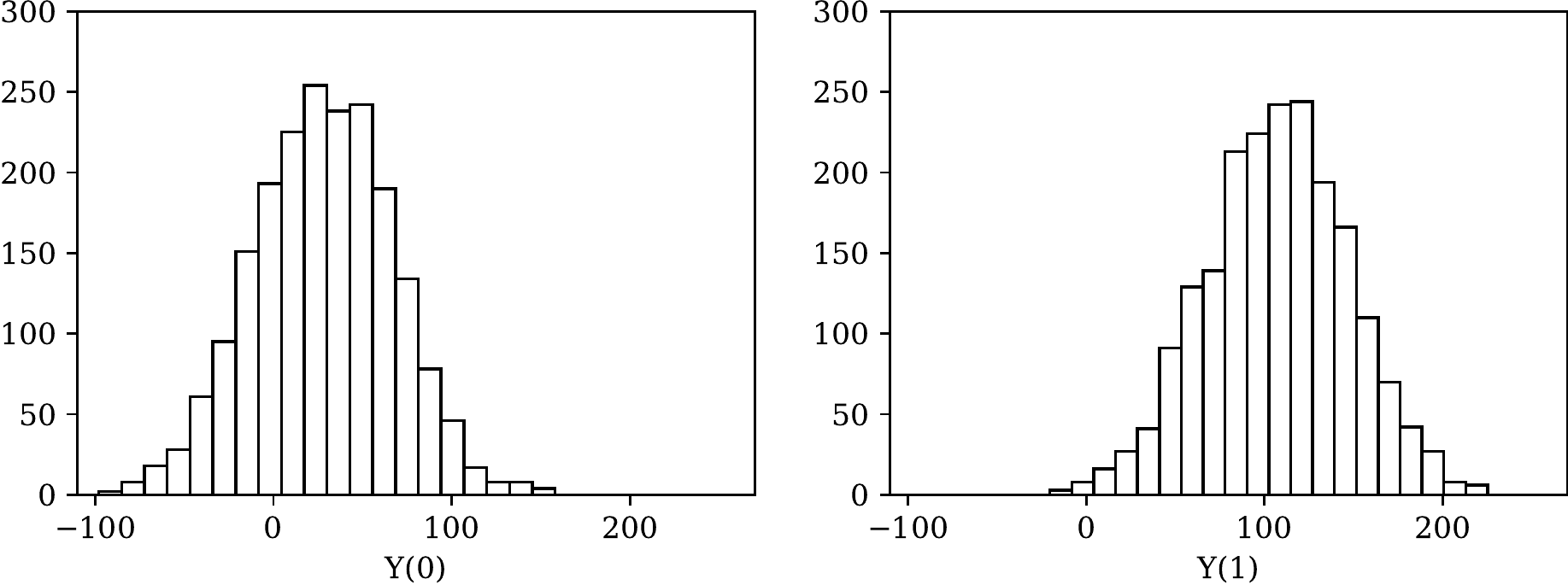}
   \caption{Histogram of potential outcomes drawn from a continuous potential outcomes framework (BART) when the true potential outcomes are counts (top panel)}
   \label{fig:Ng2}
\end{subfigure}

\caption[]{Histogram of the potential outcomes inferred with a continuous potential outcomes using BART \citep{hill2011bayesian} when the true $(Y(0),Y(1))$ is of count type, drawn from Complex Model in (\ref{eq:simple-complex-model}) with lognormal-Poisson potential outcomes (overdispersion parameter $\epsall=\log N(0,\sigma^2)$). Note that some of the realizations go negative due to the normality assumption of the continuous potential outcomes framework (lower panel).}
\label{fig:count-gaussian-infer}
\end{figure}

Finally, Figure~\ref{fig:count-gaussian-infer} offers some illustrative diagnostics that a naive evaluation of average treatment effect using a continuous potential outcome model when they are in fact of count type may lead to nonsensical answers. As can be seen from lower panel of Figure~\ref{fig:count-gaussian-infer}, some potential outcomes are negative due to the incompatibility of using a continuous real-valued potential outcomes framework to model counts. One possible avenue is to truncate the negative potential outcomes to zero, but it is unclear how this rounding should be handled in a principled manner and whether it influences the posterior predictive distribution when the true potential outcomes are of count type.

\subsection{LaLonde job training data}
The data that we use to illustrate the approximation methods developed comes from the well-known randomized evaluation program where the authors examined the effectiveness of a job training program (the treatment) on individual earnings (the outcome) in 1978. The National Supported Work (NSW) program was first analyzed by \cite{lalonde}. The sample data set used here is that of the Dehejia-Wahha sample \citep{DehejiaWahba:1999}, which consists of 445 observations, corresponding to 185 treated and 260 control subjects. For each person, we have information on their characteristics, namely age, years of education, whether they were currently or ever before married, whether they possess high school diploma, and ethnicity. We match the data, using functions in the \texttt{matchit} package \citep{matchit}, and after matching there are $370$ observations. 

\begin{table}[!ht]
\centering
\setlength{\tabcolsep}{11pt} 
\vspace{3mm}
\begin{tabular}{@{}cccccc@{}}
\toprule
\vspace{3mm}
Unit & \multicolumn{2}{c}{Potential Outcomes} & Treatment            & \begin{tabular}[c]{@{}c@{}} Observed\\ Outcome\\ (actual)\end{tabular} & \begin{tabular}[c]{@{}c@{}}Observed \\ Outcome\\(integer valued)\end{tabular} \\
\vspace{5mm}
& $Y_i(0)$ & $Y_i(1)$ & $W_i$ & $Y^{\text{obs,continuous}}_i$ & $Y^{\text{obs,categorized}}_i$\\ 
\cmidrule(lr){2-3}\cmidrule(lr){1-1}\cmidrule(lr){4-4}\cmidrule(lr){5-5}\cmidrule(lr){6-6}\\
\vspace{3mm}
1 & 0    & \textbf{?}    & 0 & 0    & 1 \\
\vspace{3mm}
2 & \textbf{?}    & 9.9  & 1 & 9.9  & 2 \\ 
\vspace{3mm}
3 & 12.4 & \textbf{?}    & 0 & 12.4 & 3 \\ 
\vspace{3mm}
4 & \textbf{?}    & 3.6   & 1 & 3.6   & 1 \\ 
\vspace{3mm}
\vdots & \vdots & \vdots & \vdots & \vdots & \vdots\\
\multicolumn{1}{l}{} & \multicolumn{1}{l}{} & \multicolumn{1}{l}{} & \multicolumn{1}{l}{} & \multicolumn{1}{l}{}                                                   & \multicolumn{1}{l}{}\\
\bottomrule
\end{tabular}
\vspace{3mm}
\caption{A subset of the data from the NSW evaluation. Note that the last column is obtained by mapping a value of 1 to earnings (real earnings in 1978) between \$0-\$5000, 2 for \$5000-\$10000, etc. The question marks represent missing potential outcomes.}
\label{tab:lalonde}
\end{table}


Drawing causal inference typically involves sensitive information that needs to be kept private (e.g., salary). Suppose that it may be preferable to assign integer values in place of exact earnings. To demonstrate the suitability of our method, we assign the integer values to the duplicated wage data. The outcome of interest is post-program labour earnings in 1978 (\texttt{re78}). We assign a value of $1$ to earnings between \$0 -- \$5,000, 2 for \$5,000 -- \$10,000, $3$ for \$10,000 -- \$15,000, etc. The maximum integer valued from this data set is 13, which corresponds to unit NSW132 with a post-program labour earnings of \$60,308. The column $Y^{\mathrm{obs,categorized}}_i$ in Table~\ref{tab:lalonde} illustrates an excerpt of the mapped integer values to actual wage earnings. The final column presents the categorized observed outcome which we used in our analysis. The integer valued data are chosen such that its distributional properties roughly depict that of a Poisson random variable, i.e. when $\epsilon\equiv 1$ in equation (\ref{eq:model}) for both the treated and control groups. The empirical mean and variance are approximately equal to $1.89$ and $1.98$ for the control group and $1.62$ and $0.97$ for the treated group, respectively.


\begin{table}[t!]
\centering
\setlength{\tabcolsep}{15pt}
\begin{tabular}{ccccc }
	\toprule
	Group $(g)$ & Salary ('000) & ATE$_{g}$ & \#Units$_{g}$ & Weights$_{g}$\\
	\cmidrule(lr){1-1}\cmidrule(lr){2-2}\cmidrule(lr){3-3}\cmidrule(lr){4-4}\cmidrule(lr){5-5}
	1 & $[0\textcolor{white}{0};5\textcolor{white}{0}]$ & $\textcolor{white}{0}0.38$ & $263$  & $0.591$ \\
	2 & $[5\textcolor{white}{0};10]$ & $\textcolor{white}{+}0.18$ & $99$ & $0.222$  \\
	3 & $[10;15]$ & $-0.07$ & $52$ & $0.117$   \\
	4 & $[15;20]$ & $\textcolor{white}{+}0.23$ & $18$ & $0.040$   \\
	5 & $[20;25]$ & $\textcolor{white}{+}0.54$ & $7$ & $0.016$   \\
	6 & $[25;30]$ & $\textcolor{white}{+}4.24$ & $2$ & $0.004$       \\
	7 & $[30;35]$ & $\textcolor{white}{+}5.23$ & $1$ & $0.002$    \\
	8 & $[35;40]$ & $-0.20$ & $2$ & $0.004$    \\
	9 & $[40;45]$ & ----- & ----- & -----         \\
	10 & $[45;50]$ & ----- & ----- & -----         \\
	11 & $[60;65]$ & $\textcolor{white}{+}9.38$ & $1$ & $0.002$        \\
	\bottomrule
\end{tabular}
\caption{Group-wise average treatment effects for data from the NSW evaluation using real wages of 1978. Note that there are no units in Groups $9$ and $10$.}
\label{tab:groupwise-1978} 
\end{table}

Using Procedure 2 with $\sigma_{\beta}=1000$, the posterior mean of the average treatment effect is $0.41$, and the posterior standard deviation equal to $0.08$. This average treatment effect translates to approximately a salary increase of $\$2050 \,\,(i.e., 0.41\times 5000)$. As a comparison, we compute the average treatment effect on this data set using the model-based inference that assumes a bivariate normal distribution for the potential outcomes given covariates (confer e.g. Chapter 8.7 in \cite{ImbensGuidoRubin:2015}). We report that the Bayesian average treatment effect for the salary increase is approximately \$1962. Although the modeling assumptions are different, this indicates that our framework behaves reasonably and has corresponding accuracy to the standard classical causal model whose potential outcomes are modeled by a bivariate normal distribution. It is important to note that even though the average treatment effect is of a similar magnitude obtained under different assumptions, the posterior distribution of $\Ymis$ estimated from the bivariate normal distribution potential outcomes model can in fact take negative values, as explicated in Section \ref{sec:synthetic}. With regards to the stability of the average treatment effects for different groups relative to the posterior mean value of $0.41$, we note that there exists heterogeneity across different groups with different average treatment effects. In Table~\ref{tab:groupwise-1978}, the columns \#Units$_{g}$ and Weights$_{g}$ represent the number of participants in each group as well as their corresponding weights, calculated as a ratio of \#Units$_{g}$ to the total number of participants, respectively. 




\begin{table}[t!]
\setlength{\tabcolsep}{7pt}
\centering
\begin{tabular}{ ccccc }
\toprule
Group $(g)$ & Salary ('000) & Estimated ATE$_{g}$ & Number of units$_{g}$ & Weights$_{g}$\\
\cmidrule(lr){1-1}\cmidrule(lr){2-2}\cmidrule(lr){3-3}\cmidrule(lr){4-4}\cmidrule(lr){5-5}
1 & $< -20.0$ & $\textcolor{white}{+}6.96$ & $6$ & $0.013$ \\
2 & [$-20.0 ; -17.5$] & $\textcolor{white}{+}9.88$ & $2$ & $0.004$ \\
3 & $[-17.5 ; -15.0]$ & $\textcolor{white}{+}8.41$ & $2$ & $0.004$ \\
4 & $[-15.0 ; -12.5]$ & $\textcolor{white}{+}2.63$ & $3$ & $0.007$ \\
5 & $[-12.5 ; -10.0]$ & $-5.95$ & $5$ & $0.011$ \\
6 & $[-10.0;\,\, -7.5]$ & $-1.28$ & $8$ & $0.018$ \\
7 & $[-7.5\,; -5.0]$ & $\textcolor{white}{+}2.37$ & $12$ & $0.027$ \\
8 & $[-5.0\,; -2.5]$ & $-0.29$ & $13$ & $0.029$ \\
9 & $[-2.5\,\, ; \textcolor{white}{+}0.0]$ & $\textcolor{white}{+}0.52$ & $17$ & $0.038$ \\
10 & $\textcolor{white}{+}[0.0\,\,; \textcolor{white}{+}2.5]$ & $\textcolor{white}{+}0.56$ & $170$ & $0.38$ \\
11 & $\textcolor{white}{+}[2.5\,\, ;\,\,\,\,5.0]$ & $\textcolor{white}{+}0.23$ & $57$ & $0.128$ \\
12 & $\textcolor{white}{+}[5.0\,\, ;\,\,\,\,7.5]$ & $-0.05$ & $45$ & $0.101$ \\
13 & $\textcolor{white}{+}[7.5\,\,\,\,;\,\,\,10.0]$ & $\textcolor{white}{+}0.31$ & $41$ & $0.092$ \\
14 & $\textcolor{white}{+}[10.0 \,\,;\,\,\,\,12.5]$ & $-0.06$ & $26$ & $0.058$ \\
15 & $\textcolor{white}{+}[12.5 \,\,;\,\,\,\,15.0]$ & $-0.40$ & $14$ & $0.031$ \\
16 & $\textcolor{white}{+}[15.0 \,\,;\,\,\,\,17.5]$ & $-0.32$ & $7$ & $0.016$ \\
17 & $\textcolor{white}{+}[17.5 \,\,;\,\,\,\,20.0]$ & $\textcolor{white}{+}4.09$ & $6$ & $0.013$ \\
18 & $> 20.0$ & $\textcolor{white}{+}3.16$ & $11$ & $0.025$ \\
\bottomrule
\end{tabular}
\caption{Estimated Group-wise average treatment effects for data from the NSW evaluation using the difference between the two wages in the years 1974 and 1978.}
\label{tab:groupwise-increase}
\end{table}


We perform another analysis that measures the different levels of incremental earnings between 1974 and 1978, in order to compare and contrast with our above mentioned observations of the average treatment effects in a single year of 1978. In this complementary study, we assign a value of $1$ to an incremental earnings difference of lesser than -\$20,001, $2$ for incremental earnings difference between [-\$20,000 ; -\$17,500], $3$ for incremental earning difference between [-\$17,500 ; -\$15,000], $4$ for [-\$15,000 ; -\$12,500], until we reach the value of $18$ for positive incremental earning difference greater than +\$20,000. The posterior mean of the average treatment effect is $0.68$, and the posterior interval equal to $0.18$. This average treatment effect translates to approximately a salary increase of $\$1700 \,\,(0.68\times 2500)$. The group wise ATE is summarized in Table~\ref{tab:groupwise-increase}.


\subsection{COVID-19}
The coronavirus disease outbreak first identified in China at the end of 2019 has become a pandemic in 2020. This virus has now spread around the world and in the United States alone, it has caused innumerable hospitalizations and daily deaths \citep{Reuters1:2020}. The current pandemic presents a greater risk for counties with higher numbers of confirmed cases, as they generally report more deaths. However, the number of deaths can be misleading as it does not account for the differences in population. Do counties with more confirmed cases also suffer from higher death rates? In this section, we initiate our proposed framework to validate that the confirmed cases do exhibit a relationship to fatality rates across different counties in the United States. 

The data we used for our analysis were obtained from the repository associated to a recent work by  \cite{Dominici:2020}. Using the COVID-19 counts of death and confirmed cases for each county up until 4 April 2020, we limit our data set to counties with confirmed cases above a certain threshold so as to reduce the uncertainties caused by smaller populations. The county population is obtained from the census.gov (https://www.census.gov/data/datasets). 

Define $c_i$ and $p_i$ to be the number of confirmed cases and population size for county $i$, respectively. As mentioned earlier, we focus on those counties that have expected number of COVID-cases cases per 100,000 that exceeds $1$. Furthermore, the county-wise crude death per capita  $\texttt{cdr}_i:=10^q \times d_i/p_i$ is computed,
where $d_i$ denotes the death counts for county $i$. Similarly, the crude number of confirmed cases per capita in county $i$ is computed via the following expression $\texttt{ccr}_i:=10^q \times c_i/p_i$. Throughout our analysis, we take $q=5$. Define the treatment indicator $W_i$ to take a value equals to $1$ if $\texttt{ccr}_i \geq h$ (large number expected confirmed cases per 100,000) and 0 otherwise (low magnitude). Here, we set $h\equiv 70$ and we remark that this quantity can be varied accordingly. With this, we can associate to each county $i$ two potential outcomes, namely $Y_i(1)$, the crude death rate in county $i$ had the \emph{exposure} in county $i$ been larger than $h$, and $Y_i(0)$ otherwise. The observed crude mortality rate in county $i$ is denoted by $Y^{\mathrm{obs}}_i$. We refer to the counties with $W_i=1$ as `infected counties' and to the counties with $W_i=0$ as `less-infected counties' during the study period. The sample size under consideration is $399$ with $181$ treated and $218$ control counties.

We consider the following covariates:
poverty ($\texttt{poverty}$) and education levels ($\texttt{education}$),
percentage of owner-occupied housing ($\texttt{pct\_owner\_occ}$), percentage of Hispanic people ($\texttt{hispanic}$), percentage Asian people ($\texttt{pct\_asian}$), percentage of native population ($\texttt{pct\_native}$), percentage of white people ($\texttt{pct\_white}$) and percentage of black population ($\texttt{pct\_blk}$). Details on these variables can be found in \cite{Dominici:2020}.

\begin{figure}[t]
\centering
  \begin{subfigure}[b]{0.48\textwidth}
    \includegraphics[width=\textwidth]{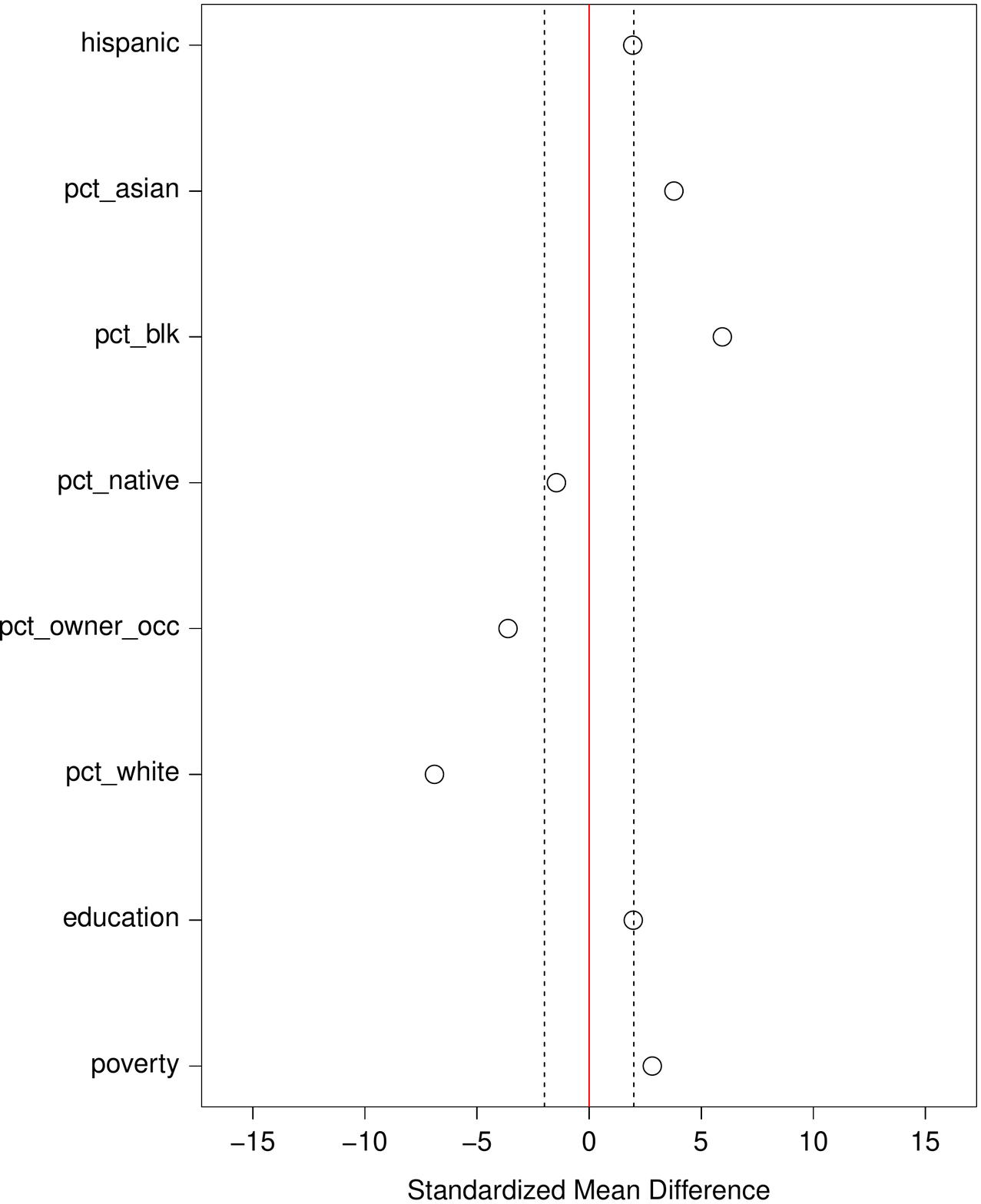}
    \caption{Before matching}
    \label{fig:1}
  \end{subfigure}
  \begin{subfigure}[b]{0.48\textwidth}
    \includegraphics[width=\textwidth]{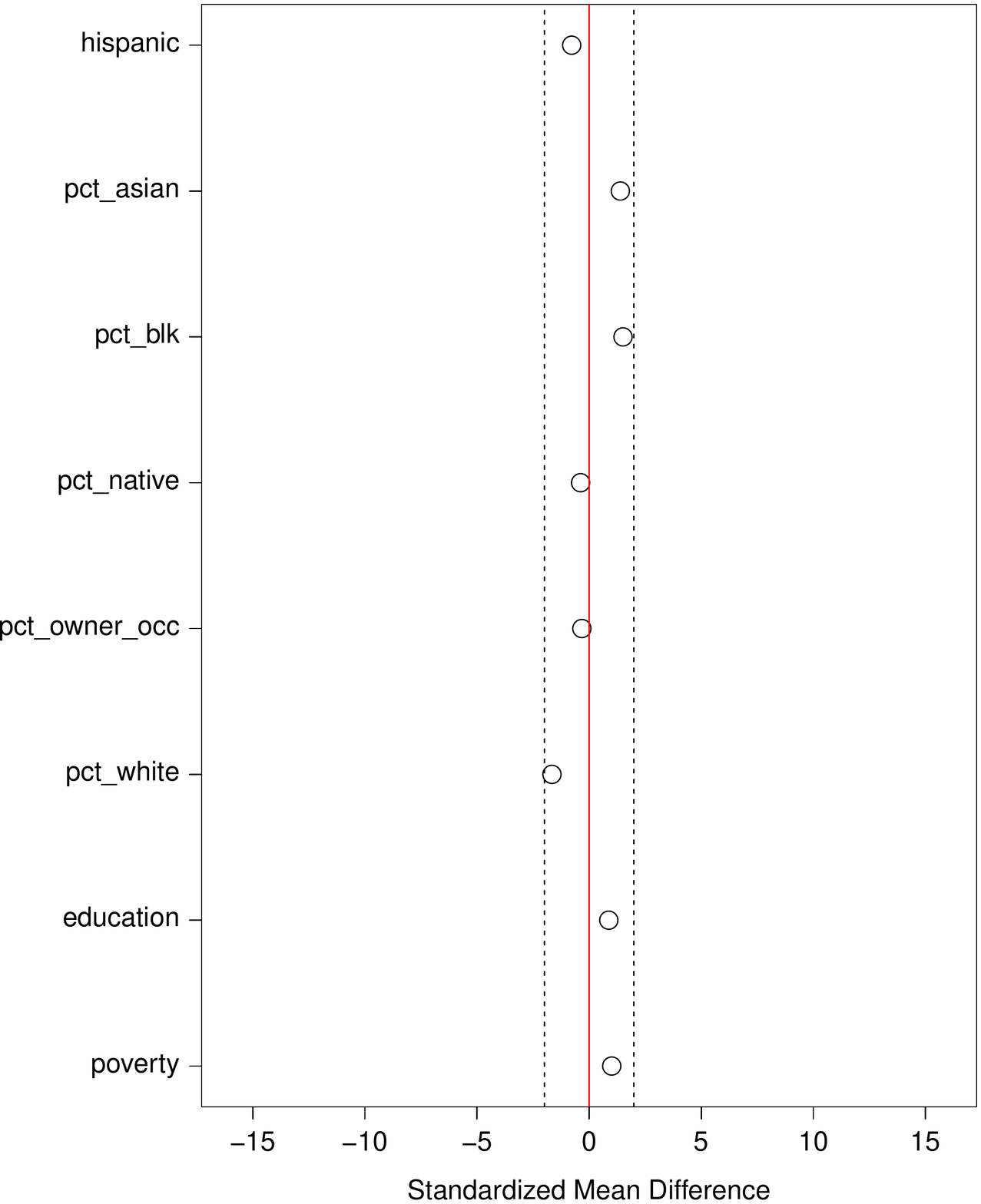}
    \caption{After matching}
    \label{fig:2}
  \end{subfigure}
  \caption{Covariate balance before and after matching.}
  \label{fig:before-after-matching}
\end{figure}

We perform a one-to-one matching strategy by imposing a tolerance level on the maximum estimated propensity score distance (caliper) to create groups of exposed and control days with similar distributions for covariates \citep{RR:1983}. The Love plots (e.g., \cite{Ahmed:2007}) are shown in Figure~\ref{fig:before-after-matching}, where they schematically display covariate balance before and after matching. We see from Figure~\ref{fig:before-after-matching}(a) that the covariates are imbalanced, especially $\texttt{pct\_white}$ and $\texttt{pct\_blk}$. The difference in means of covariates is standardized using the expression ${{(\theta_1-\theta_0)}/\sqrt{{v_0/n_0+v_1/n_1}}}$ where $\theta_0$ $(v_0)$ and $\theta_1$ $(v_1)$ denote the average (variance) of the corresponding control and treated covariates with $n_0$ and $n_1$ being the number of control and treated counties, respectively. 



With the parameter settings of $\sigma_{\beta}=100$ and $\epsall=1$, the posterior mean of the average treatment effect is $172$ deaths from COVID-19 per 100,000 and its posterior standard deviation is $\sim 1.0$. These results suggest that on average, more deaths occur in counties when there are more cases ($\texttt{ccr}>h\equiv 70$). Our choice of $\sigma_{\beta}=100$ is intended to establish little or no
structure through our model assumptions. We further note that $h$ can be varied to analyze different scenarios given the data, which at the present moment may still be limited and very incomplete.


\section{Postlude}
\label{sec:postlude}
In this paper, we have proposed a framework for estimating causal inference in a Bayesian setting when the potential outcomes are counts. Under this paradigm, standard causal models that handle continuous or binary potential outcomes (e.g., \cite{Rubin2006,hill2011bayesian,GutmanRubin:2012,Gutman-Rubin:2015}) would not be directly applicable since the potential outcomes can take non negative integer values. Presented within the Rubin 
causal framework, we argued that imputing the missing count potential outcomes dominates the commonly-used approach that directly regresses the count outcomes on the observed treatment and background covariates by allowing flexibility in drawing inferences.

Our proposed framework provides a Bayesian framework that can be extended to address more complex causal questions with count data. Some examples are: $(i)$ Principal stratification in non-compliance settings, $(ii)$ Potential count time series, $(iii)$ zero-inflated potential outcomes,  $(iv)$  nonparametric modeling of covariates with count potential outcomes, and $(v)$ multifactorial designs with count outcomes. Therefore, it is envisaged that this article will open up a discussion on causal problems whose data comes in the form of counts, which can be found in almost all areas of statistics, health, social, and physical sciences.

\vspace{5mm}

\newpage

\bibliographystyle{apa}
\bibliography{count-data-potential-outcomes.bib}

\end{document}